\documentclass[jmp,11pt,tightenlines,nofootinbib,showpacs,longbibliography,notitlepage]{revtex4-1} 

\usepackage{amsmath,amsthm,amssymb,amsfonts,mathrsfs,url}
\usepackage{bbm}
\usepackage{enumerate}
\usepackage[utf8]{inputenc}

\usepackage[usenames,dvipsnames]{color}
\usepackage{hyperref}

\hypersetup{
    pdfkeywords={spinfoam, spin foam, loop quantum gravity, general relativity,
                 coherent states, Regge calculus, semiclassical limit}, 
    colorlinks=true,         
    linkcolor=blue,          
    citecolor=blue,        
    urlcolor=Violet,      
}

\newcommand{\ds}{\displaystyle}

\newcommand{\calg}{C$^*$-algebra}
\newcommand{\cs}{C$^*$}

\newcommand{\D}{\operatorname{\mathscr D}}
\newcommand{\Sch}{\operatorname{\mathscr S}}
\newcommand{\Neigh}{\operatorname{\mathcal N}}
\newcommand{\Aut}{\operatorname{Aut}}
\newcommand{\Diff}{\operatorname{Diff}}
\newcommand{\End}{\operatorname{End}}

\newcommand{\Mat}{\operatorname{Mat}}
\newcommand{\Ran}{\operatorname{Ran}}

\newcommand{\Tr}{\operatorname{Tr}}
\newcommand{\Ad}{\operatorname{Ad}}
\newcommand{\diag}{\operatorname{diag}}
\newcommand{\de}{\text{\upshape d}}

\newcommand{\Id}{\operatorname{Id}}
\newcommand{\supp}{\operatorname{supp}}
\newcommand{\Span}{\operatorname{Span}}
\newcommand{\deriv}[4][]{\frac{{#2}^{#1}#3}{{#2}{#4}^{#1}}}

\newcommand{\pder}[3][]{\deriv[#1]{\partial}{#2}{#3}}
\newcommand{\norm}[1]{\Vert{#1}\Vert}
\newcommand{\qc}[1]{{#1}_q\oplus{#1}_c}

\newcommand{\IR}{\mathbbm{R}}
\newcommand{\IC}{\mathbbm{C}}
\newcommand{\IE}{\mathbbm{E}}
\newcommand{\Fou}{\mathscr{F}}
\newcommand{\Bound}{\mathfrak{B}}
\newcommand{\Hilb}{\mathscr{H}}
\newcommand{\Poin}{\mathscr{P}}
\newcommand{\Lor}{\mathscr{L}}

\newcommand{\At}{\Big|}
\newcommand{\bv}[1]{\mathbf{#1}}
\newcommand{\hv}[1]{\hat{\bv{#1}}}
\newcommand{\Lup}{\Lor^\uparrow_+}
\newcommand{\Pup}{\Poin^\uparrow_+}

\newcommand{\A}[1]{\mathfrak{#1}}

\renewcommand{\Im}{\operatorname{Im}}

\newcommand{\sstd}{\sigma^{\text{std}}}
\newcommand{\Estd}{\mathcal E_{\text{std}}}

\numberwithin{equation}{section}

\theoremstyle{plain}
\newtheorem{theorem}{Theorem}\numberwithin{theorem}{section}
\newtheorem{proposition}[theorem]{Proposition}
\newtheorem{corollary}[theorem]{Corollary}
\newtheorem{lemma}[theorem]{Lemma}
\newtheorem{definition}[theorem]{Definition}

\begin{document}

	\title{A scale-covariant quantum space-time}
	\author{Claudio Perini}
    \email      {claude.perin@gmail.com}
    \affiliation{Institute for Gravitation and the Cosmos, Physics Department,
                 Penn State, University Park, PA 16802-6300, USA}
\author{Gabriele Nunzio Tornetta}
    \email      {g.tornetta.1@research.gla.ac.uk}
    \affiliation{School of Mathematics and Statistics, University of Glasgow,
                 15 University Gardens, G12 8QW, Scotland}

	\begin{abstract}
	A noncommutative space-time admitting dilation symmetry was briefly mentioned in the seminal work \cite{doplicher95} of Doplicher, Fredenhagen and Roberts.
In this paper we explicitly construct the model in details and carry out an in-depth analysis. The \calg\ that describes this quantum space-time is determined, 
and it is shown that it admits an action by $*$-automorphisms of the dilation group, along with the expected Poincaré covariance.
In order to study the main physical properties of this scale-covariant model, a free
scalar neutral field is introduced as a investigation tool. Our key results are then the
loss of locality and the irreducibility, or triviality, of special field algebras
associated with regions of the ordinary Minkowski space-time. It turns out, in the conclusions, that this analysis allows also to argue
on viable ways of constructing a full conformally covariant model for quantum
space-time.

	\end{abstract}
	
	\maketitle
	
	
	\section{Introduction}
In this paper we study a non-commutative space-time of DFR-type, that can be obtained as the limiting scale-free case of the original DFR model \cite{doplicher95}. The main mathematical interest to study this model is that it is Poincaré \emph{and} dilation covariant, thus it possesses almost all the symmetries given by the conformal group. The issue of implementing the remaining symmetry, which is the relativistic ray inversion, is discussed in the concluding section.

The paper is organized as follows. The basic model of quantum space-time is reviewed in the subsequent section \ref{chap:basic model}. The scale-covariant model is introduced and analyzed in section \ref{chap:model}, and in section \ref{chap:field} a free neutral scalar field is used to test some properties of the scale-covariant quantum space-time.

Let us review in this introduction the main motivation behind the DFR quantum space-time model, here called the basic model. 

Classical general relativity is a well-established theory supported by practically all experiences. But the concurrence of its principles with the basic principles of quantum mechanics point to difficulties at small length scales. If we attach an operational meaning to space-time events to within a desired accuracy, a break-down of the theory is expected to occur at a very short scale.

The idea that at small distances there must be limitations on the localizability of spacetime events is very old and can be traced back to Wigner and Salecker \cite{Wigner:1957ep, Salecker:1957be}. This idea has been revived in \cite{doplicher95} to motivate the introduction of space-time uncertainties. Let us suppose that we are interested in measuring the position of an object to within the accuracy $\Delta x$. Then, according to Heisenberg's principle of indeterminacy, an uncontrolled momentum $\Delta p$ such that $\Delta p \Delta x\gtrsim\hbar$ is involved in the measuring process. Using the constant $c$, i.e. the speed of light in the vacuum, we may associate to this momentum an energy $\Delta E$ such that
\begin{equation*}
	\Delta E\Delta x\approx \hbar c.
\end{equation*}
If we make the further assumption that this uncontrolled energy $\Delta E$ will be spherically distributed in space, we can easily compute its Schwarzschild radius, namely
\begin{equation}
	r_S=\frac{2G\Delta E}{c^4}.
\end{equation}
Now, if the Schwarzschild radius associated to this measurement is larger than the required accuracy $\Delta x$, the event is hidden from a distant observer and no operational meaning can be attached to it anymore. Thus the break-down ought to occur whenever $\Delta x$ becomes comparable to $r_S$, namely
\begin{equation}
	\Delta x\approx r_S\approx\lambda_P,
\end{equation}
where $\lambda_P$ is the Planck's fundamental length,
\begin{equation}
	\lambda_P = \sqrt{\frac{G\hbar}{c^3}}.
\end{equation}
This difficulty is avoided if we postulate the following principle of gravitational stability against localization of events, which states that
\begin{quote}
\emph{the gravitational field should not be so strong to prevent the event to be seen from a distant observer - distant compared to the Planck scale.}
\end{quote}
Non-commutative geometry can be viewed as a way to address the issue of the localization of events by replacing a given classical space-time with a non-commutative space-time, where quantum-mechanical uncertainty relations between space-time events hold. Space-time uncertainty relations prevent the sharp localization of events, since the points become \emph{fuzzy}, or \emph{blurred}, at very short scales, and preserve gravitational stability against measurement.

Being motivated by the issue of localizability, which merges classical gravity with quantum physics, non-commutative geometry is certainly an aspect of quantum gravity. However it does not usually follow a general relativistic perspective, since the space-time which is quantized is \emph{fixed} and not dynamical as in the classical gravitational theory. It is possible that the gravitational stability can be derived from basic principles in a theory of quantum gravity where gravity and quantum physics are truly unified and quantum space-time is dynamical \cite{rovelli04,oriti09}.

In the non-commutative geometry approach, a suitable expression for the uncertainty relations on Minkowski space-time can be derived explicitly, and leads to the DFR basic model \cite{doplicher95}. Using the Einstein's field equations linearized around the flat solution, semi-classical considerations and the principle of gravitational stability against localization of events,  the following uncertainty relations have been derived,
\begin{subequations}\label{eq:coord_uncert}
	\begin{align}
		\Delta x_0\sum_{k=1}^3\Delta x_k&\gtrsim\lambda_P^2,\\
		\sum_{i<k=1}^3\Delta x_i\Delta x_k&\gtrsim\lambda_P^2,
	\end{align}
\end{subequations}
where $\Delta x_\mu$ is the uncertainty associated to the Cartesian coordinate $x_\mu$ of an event in Minkowski space-time.

In turn, the uncertainty relations do arise from quantum-mechanical operators. Indeed it was shown that if the Minkowski coordinates are promoted to self-adjoint operators, subjected to a natural set of Lorentz-invariant quantum conditions, they lead precisely to the uncertainty relations \eqref{eq:coord_uncert} above (see \cite{doplicher95} and next section). Non-commuting coordinates, and their generated operator algebra, are the basis of the mathematical description of non-commutative space-times. Non-commutativity forbids arbitrarily accurate simultaneous measurements of all the coordinates of an event. Thus we expect a `point' in quantum space-time to look like a \emph{blurred dot} of characteristic linear dimensions given by Planck's length $\lambda_P$. This \emph{unfocused} view of space-time should inevitably lead to the loss of the property of locality of quantum field theories on quantum space-time, and we shall see in details (cf. sec. \ref{sec:cov_and_loc}) that this is indeed the case.

	\section{The basic model of quantum space-time\label{chap:basic model}}
This section is dedicated to a review of the Doplicher-Fredenhagen-Roberts (DFR) basic model of quantum Minkowski space-time \cite{doplicher95}. We will not give all the mathematical proofs, that can all be found in \cite{doplicher95}. This section will serve as an introduction and reference to the scale-covariant model analyzed in section \ref{sec:scale-covariant}.

The first step for the quantization is the introduction of unbounded self-adjoint operators $\{q_\mu|\mu=0,\ldots,3\}$ which represent the quantum space-time coordinates. As mentioned above, the uncertainty relations come from the ansatz that these operators do not commute, i.e.
\begin{equation}\label{eq:cr}
		[q_\mu, q_\nu] \subset i\lambda_P^2Q_{\mu\nu},
\end{equation}
where the dimensionless operators $Q_{\mu\nu}$ are the self-adjoint closures of the commutators\footnote{or more correctly of  $i\lambda_P^{-2}[q_\nu,q_\mu]$.}, and from other algebraic conditions given below. 

For simplicity reasons, the commutators $Q_{\mu\nu}$ are assumed to commute with all the coordinate operators $q_\mu$, i.e. they are the generators of the center of the associated Lie algebra.
Moreover non-commutativity is evaluated by the operator
	\begin{equation}\label{eq:pseudoscalar}
		-\frac12Q_{\mu\nu}(*Q)^{\mu\nu},
	\end{equation}
which is of course proportional to $\epsilon^{\mu\nu\alpha\beta}q_\mu q_\nu q_\alpha q_\beta$, and so vanishing in the commutative case. The above object \eqref{eq:pseudoscalar} is invariant under the most general proper Poincar\'e
transformation
	\begin{equation}
		q\mapsto\Lambda q + a\Id,
	\end{equation}
but it is not invariant under time or space reflection, i.e. the improper Poincar\'e
group, for it behaves like a pseudo-scalar. On the other hand, the operator
	\begin{equation}\label{eq:invariant}
		\frac12Q_{\mu\nu}Q^{\mu\nu}
	\end{equation}
is a genuine full Lorentz scalar.

Recalling that a general skew-symmetric 2-tensor on $\IR^4$ can be represented by
means of two vectors $\bv e,\bv m\in\IR^3$, namely its electric and magnetic component
respectively, in analogy with the electromagnetic 2-form of electrodynamics, we can express the components
of a generic skew-symmetric tensor $\theta$ on $\IR^4$ in matrix notation by
	\begin{equation}
		\theta_{\mu\nu} = \begin{bmatrix}0&\bv e^T\\-\bv e&*\bv m\end{bmatrix},\qquad \bv e,\bv m\in \IR^3
	\end{equation}
where the superscript $T$ denotes matrix transposition and $*\bv m$ is the Hodge-dual
of $\bv m$ in $\IR^3$. This leads to the equivalent expressions
\begin{subequations}\label{eq:scalar_and_pscalar}
	\begin{align}
		-\frac12Q_{\mu\nu}(*Q)^{\mu\nu} &= \bv e\cdot\bv m + \bv m\cdot\bv e,\\
		\frac12Q_{\mu\nu}Q^{\mu\nu} &= \norm{\bv m}^2-\norm{\bv e}^2,
	\end{align}
\end{subequations}
for the operators \eqref{eq:pseudoscalar} and \eqref{eq:invariant} respectively. The basic model of quantum space-time arises by requiring symmetry in both $\bv e$ and $\bv m$ and it is then defined by the following quantum conditions:
\begin{subequations}\label{eq:basic}
	\begin{align}
		\frac12Q_{\mu\nu} Q^{\mu\nu}&=0,\\
		\label{eq:basic:second}\left(\frac12Q_{\mu\nu} (*Q)^{\mu\nu}\right)^2&=\Id,\\
		[q_\lambda,Q_{\mu\nu}] &= 0\qquad\forall\lambda,\mu,\nu=0,\ldots,3.
	\end{align}
\end{subequations}
The first one expresses the sought symmetry in $\bv e$ and $\bv m$. The second condition sets the characteristic length scale of non-commutativity to $1$ in natural units, that is to $\lambda_P$ in generic units. The third additional condition is just a statement of the centrality of the $Q$s.

It has been shown in \cite{doplicher95} that the quantum conditions \eqref{eq:basic} imply the uncertainty relations
\begin{subequations}
	\begin{align}
		\Delta_\omega q_0\sum_{k=1}^3\Delta_\omega q_k&\geq\frac12\lambda_P^2,\\
		\sum_{i>k=1}^3\Delta_\omega q_i\Delta_\omega q_k&\geq\frac12\lambda_P^2,
	\end{align}
\end{subequations}
where $\omega$ is any state in the domain of the operators $[q_\mu,q_\nu]$.

\subsection{Irreducible representations of the quantum coordinates}

So far we haven't specified any representation of the above operators on a Hilbert space. Here we shall review the irreducible representations of the quantum coordinates for the DFR basic model. Throughout this section we will work with a unit system such that $\lambda_P=1$.

In the spirit of von Neumann's proof of uniqueness of the Schr\"odinger's representation, we shall only consider \emph{regular} representations of the relations \eqref{eq:basic}, i.e. those that can be `integrated' into Weyl's
form
	\begin{equation}\label{eq:weyl}
		U(\alpha)=e^{i\alpha^\mu q_\mu},\qquad\alpha\in\IR^4
	\end{equation}
and that satisfy the Weyl's commutation relations
	\begin{equation}\label{eq:weyl_relations}
		U(\alpha)U(\beta)=e^{-\frac i2\sigma(\alpha,\beta)}U(\alpha+\beta),\qquad
			\sigma(\alpha,\beta)=\alpha^\mu Q_{\mu\nu}\beta^\nu.
	\end{equation}

Notice that \eqref{eq:weyl} is a sort of quantum plane wave, and \eqref{eq:weyl_relations} gives the rule for multiplying plane waves. Concerning the irreducible representations of \eqref{eq:weyl_relations}, a key role is played by the operators $Q_{\mu\nu}$, for since they are supposed to be central, Shur's lemma implies that for an irreducible representation they are proportional to the identity operator, namely
	\begin{equation}\label{eq:Q_irreducible}
		Q_{\mu\nu}=\sigma_{\mu\nu}\Id,
	\end{equation}
where $\sigma$ is a skew-symmetric 2-tensor on the usual Minkowski space-time, and its
electric and magnetic components satisfy the numerical counterpart of the conditions
\eqref{eq:basic}. Thus in an irreducible representation eq. \eqref{eq:cr} becomes
\begin{equation}\label{eq:cr_irreducible}
		[q_\mu, q_\nu] \subset i\sigma_{\mu\nu}\Id,
	\end{equation}
and the Weyl's relations \eqref{eq:weyl_relations} hold, with $Q$s replaced by the corresponding $\sigma$s. The claim that $\sigma$ is a tensor is justified by the properties of the joint spectrum $\Sigma$ of the operators $Q_{\mu\nu}$.  As a manifold, $\Sigma$ is homeomorphic to $TS^2\times\mathbb \{-1,+1\}$, i.e. the disjoint union of two tangent bundles of the 2-sphere. Further properties of $\Sigma$ are
related to its behavior under the action of the full Lorentz group given by
\begin{subequations}\label{eq:action_on_jsp}
	\begin{align}
		\Lor\times\Sigma &\to \Sigma\\
		(\Lambda,\sigma)&\mapsto \Lambda\sigma\Lambda^T.
	\end{align}
	\end{subequations}
It turns out that $\Sigma$ is a homogeneous space under this action, and setting
$\Sigma_\pm := TS^2\times\{\pm1\}$, then the action of the proper orthochronous Lorentz
subgroup $\Lup$ on each connected part $\Sigma_\pm$ is transitive.

Now observe that, since the Lorentz group acts transitively on $\Sigma$ through the action previously given,
it is possible to fix a special point $\sigma_0\in\Sigma$ and link any desired point $\sigma$ to it by means of a suitable Lorentz transformation
$\Lambda_\sigma\in \Lor$ such that 
	\begin{equation}\label{eq:sigma_transform}
		\Lambda_\sigma\sigma_0\Lambda_\sigma^T=\sigma.
	\end{equation}
A suitable choice in the case of the Doplicher-Fredenhagen-Roberts basic model is $\bv e=\bv m=(0,1,0)$, i.e. the
skew-symmetric block matrix
	\begin{equation}
		\sigma_0 = \begin{bmatrix}0 & -\Id_2\\\Id_2 & 0\end{bmatrix}.
	\end{equation}
Denoting the standard symplectic form on the plane by  $J$, namely
	\begin{equation}\label{eq:stdsympl}
		J=\begin{bmatrix}0&-1\\1&0\end{bmatrix},
	\end{equation}
then the following equivalences hold,
	\begin{equation}\label{eq:JplusJ}
		\sigma_0 \cong J\otimes\Id_2 \cong J \oplus J.
	\end{equation}
We shall denote with $q^\sigma$ the quantum coordinates in a regular irreducible representation labeled by the skew-symmetric matrix $\sigma$. Equation \eqref{eq:JplusJ} shows that we are exactly in the situation of non-relativistic quantum mechanics with a 4-dimensional phase space. Hence, by the von Neumann's uniqueness theorem it follows that all the irreducible
representations labeled by the standard matrix $\sigma_0$ are unitarily equivalent to the Schr\"odinger's representation. Explicitly, the underlying Hilbert space can be realized as $L^2(\IR)\otimes L^2(\IR)$, and the quantum coordinates are realized as
\begin{align}
q^{\sigma_0}_0&=Q\otimes \Id,\\
q^{\sigma_0}_1&=P\otimes \Id,\\
q^{\sigma_0}_2&=\Id\otimes Q,\\
q^{\sigma_0}_3&=\Id\otimes P,
\end{align}
where $Q$ is the multiplication operator by $s$ in $L^2(\IR,ds)$, and $P$ is the differentiation operator $\frac{1}{i}\frac{d}{ds}$. Although $q^\sigma$ and $q^{\sigma'}$ are not equivalent
whenever $\sigma\neq\sigma'$, if we take $\Lambda\in \Lor$ and consider the Lorentz transformation
	\begin{equation}
		q^{\sigma_0}\mapsto \Lambda q^{\sigma_0},
	\end{equation}
then the commutation relations for the irreducible representations \eqref{eq:cr_irreducible}
transform accordingly into
	\begin{equation}
		[(\Lambda q^\sigma)_\mu,(\Lambda q^\sigma)_\nu] =
			i\sigma_{\mu\nu}\Id,\qquad\sigma = \Lambda\sigma_0\Lambda^T.
	\end{equation}
Along with the transitivity of $\Lor$ on $\Sigma$, the above line implies that each irreducible
representation $q^\sigma$ must be unitarily equivalent to $\Lambda_\sigma q^{\sigma_0}$
for some $\Lambda_\sigma\in \Lor$. Hence we conclude that any irreducible representation
at $\sigma\in\Sigma$ can be expressed in terms of a Lorentz transformation of the irreducible representation at $\sigma_0\in\Sigma$.
\subsection{The \calg\ of quantum space-time and localization states\label{sec:calg_basic}}

We introduce the non-commutative Banach *-algebra $\mathcal E_0$ of functions $\mathcal C_0(\Sigma,L_1(\IR^4))$ from $\Sigma$ to $L_1(\IR^4,\text{d}^4\alpha)$ that vanish at infinity. The product is defined as
\begin{align}
		(f\times g)(\sigma,\alpha)&:=\int f(\sigma,\alpha')g(\sigma,\alpha-\alpha')e^{\frac i2\sigma(\alpha,\alpha')}\de^4\alpha'.
\end{align}
Now we are in a situation that much resembles that of ordinary quantum mechanics, namely the algebra $\mathcal E_0$ has non-degenerate representations in one-to-one correspondence with the \emph{regular} representations of the quantum conditions \eqref{eq:basic}. Thus the construction of the \calg\ describing quantum space-time may be carried out
by mathematical analogy with the C*-algebra describing the non-commutative phase space of a quantum particle.

We recall that the \calg\ of quantum spacetime is defined as the enveloping \calg\ of the Banach *-algebra $\mathcal E_0$. It is the non-commutative counterpart of the commutative C*-algebra of continuous functions on an ordinary manifold.

The result of this procedure is in the following
	\begin{theorem}[\cite{doplicher95}] \label{eq:basic_algebra} The \calg\ $\mathcal E$ of the
		basic model of quantum space-time is isomorphic to $C_0(\Sigma,\mathcal K)$,
		where $\mathcal K$ is the \calg\ of all the compact operators on a separable
		Hilbert space.
	\end{theorem}

Following \cite{doplicher95}, we shall interpret each state $\omega\in\mathcal S(\mathcal E)$
as giving localization information for events occurring in quantum space-time. The structure of the states, as well as that of the algebra $\mathcal E$, strongly
depends on the model under consideration. In the case of the basic model that we
are considering in this section, the uncertainty
	\begin{equation}
		\sum_\mu(\Delta_\omega q^{\sigma_0}_\mu)^2
	\end{equation}
attains its minimum value on the ground state of a two dimensional harmonic
oscillator, i.e. on a Gaussian wave-packet. This allows us to parametrize states
of optimal localization by means of 4-vectors in Minkowski space-time and measures $\mu$
carried by a subset of the base space of $\Sigma$, i.e. the unit sphere $\Sigma^{(1)}$,
as shown in \cite[Proposition 3.4]{doplicher95}. Explicitly, the optimal localization states take the form
	\begin{equation}\label{eq:optimal}
		\omega_{x,\mu}(f)=\int f(\sigma,\alpha)e^{i\alpha^\mu x_\mu}
			e^{-\frac12\sum_\mu\alpha_\mu^2}\ \de^4\alpha\ \de\mu(\sigma),\qquad f\in\mathcal E_0,
	\end{equation}
where we are restricting to the Banach $*$-algebra $\mathcal E_0$ (cf \cite{doplicher95}). We shall give more details on this point for the case of the scale-covariant model, which is our central interest in this paper.



\subsection{Quantum field theory on Quantum Space-time\label{sec:qft}}

A possible way of introducing calculus on quantum space-time is discussed in \cite[§5]{doplicher95}.
For the sake of the reader we give here those definitions, for they are useful in
introducing quantum field theory on quantum space-time.

Let $\{q^\mu$, $\mu=0,\ldots,3\}$ be the unbounded self-adjoint operators associated
with the coordinates of space-time, and let $f\in\Fou[L^1(\IR^4)]$, where $\Fou$ denotes the Fourier transform. The operator
$f(q)$ can be defined as
	\begin{equation}\label{eq:vnncc}
		f(q) := \int \check f(\alpha)e^{i\alpha^\mu q_\mu}\de^4\alpha,
	\end{equation}
in the spirit of von Neumann's non-commutative functional calculus.

Derivatives with respect to the space-time coordinates may be defined in terms of
the infinitesimal generators of the group of translations, which act as automorphisms
$\tau_\xi$, $\xi\in\IR^4$, via
	\begin{equation}
		\tau_\xi(f(q))=f(q-\xi\Id).
	\end{equation}
The derivative $\partial_\mu$ of a function of the quantum coordinates is then defined as
	\begin{equation*}
		\partial_\mu f(q):=\pder{}{\xi^\mu}\At_0\tau_{-\xi}(f(q)),
	\end{equation*}
i.e.
	\begin{equation}\label{eq:derivative}
		\partial_\mu f(q):=\pder{}{\xi^\mu}\At_0f(q+\xi\Id).
	\end{equation}

Following \cite{doplicher95, piacitelli11}, the above definitions allow
us to introduce a free scalar neutral field $\phi$ on quantum space-time, leading
to the usual Fock construction based on creation and annihilation operators. Letting
$\de\Omega_m^+$ denote the invariant measure on the hyperboloid of mass $m$ contained
in the future light-cone $\bar V^+$, i.e. \cite{haag96}
	\begin{equation}
		\de\Omega_m^+(\bv k)=\delta(k^2-m^2)\theta(k^0)\de^4k,
	\end{equation}
the field $\phi$ will then be defined by
	\begin{equation}\label{eq:scalar_field}
		\phi(q):=\frac1{(2\pi)^{3/2}}\int\left[e^{ik^\mu q_\mu}\otimes a(\bv k) +
			e^{-ik^\mu q_\mu}\otimes a(\bv k)^*\right]\de\Omega_m^+(\bv k),
	\end{equation}
where the first factor is an element of the \calg\ of the quantum space-time $\mathcal E$,
and the second one is an operator acting on $\Hilb_F$, the Hilbert space of the Fock
construction. As expected, $\phi(q)$ satisfies the Klein-Gordon equation, for a simple
and direct computation shows that
	\begin{equation}
		(\square + m^2)\phi(q) = 0.
	\end{equation}

There is a natural way to define a map from the states of $\mathcal S(\mathcal E)$
of the \calg\ describing the basic model of quantum space-time to operators on the
Hilbert space $\Hilb_F$, namely
	\begin{equation}\label{eq:field_on_state}
		\phi(\omega) := (\omega\otimes\Id)(\phi(q)),
	\end{equation}
for a generic state $\omega\in\mathcal S(\mathcal E)$. With a few but easy steps
we may cast the above definition in an alternative form, i.e.
	\begin{equation}
		\phi(\omega) = \int\phi(x)\psi_\omega(x)\de^4x,
	\end{equation}
where we have introduced the Fourier transform of $\omega(e^{ik^\mu q_\mu})$, namely
	\begin{equation}\label{eq:psi}
		\psi_\omega(x)=\frac1{(2\pi)^4}\int\omega(e^{ik^\mu q_\mu})e^{-ik^\mu x_\mu}\de^4k,
	\end{equation}
and where $\phi(x)$ is the usual free scalar neutral field on classical space-time
(cf. \cite[§6]{doplicher95}). Thus we have regained the familiar expression of the
scalar neutral free field smeared out with a ``test'' function, namely
	\begin{equation}
		\phi[\psi_\omega] = \int\phi(x)\psi_\omega(x)\de^4x.
	\end{equation}
For example, for a optimal localization state, the associated function $\psi_{\omega_{\xi,\mu}}(x)$
turns out to be
	\begin{equation}\label{eq:testfunc}
		\psi_{\omega_{\xi,\mu}}(x) = \frac1{(2\pi)^2}e^{-\frac12\sum_\mu(x_\nu-\xi_\nu)^2}.
	\end{equation}

The locality of such a field theory can be tested by evaluating the commutator
$[\phi(\omega_\xi),\phi(\omega_\eta)]$ on optimal localization states
$\omega_\xi$ and $\omega_\eta$. Denoting the usual commutator $[\phi(x),\phi(y)]$
for the scalar neutral free field $\phi(x)$ on classical space-time by $i\Delta(x-y)$, that is \cite{araki00}
	\begin{equation}
		i\Delta(x-y) = \frac 1{(2\pi)^3}\int e^{ip^\mu(x_\mu-y_\mu)}
			\left[\de\Omega_m^+(\bm p) - \de\Omega_m^-(\bm p)\right],
	\end{equation}
a straightforward computation shows that \cite{doplicher95}
	\begin{equation}\label{eq:commutator}
		[\phi(\omega_\xi),\phi(\omega_\eta)] = i\int\Delta(x-y)\psi_{\omega_\xi}(x)
			\psi_{\omega_\eta}(y)\ \de^4x\ \de^4y.
	\end{equation}
Although an explicit solution can be exhibited only for the massless case $m=0$ , the
massive case has basically the same behavior, as shown in \cite{doplicher95},
where the explicit form of the commutator for the massless case is given, namely
	\begin{equation}\label{eq:qcommutator}
		[\phi(\omega_\xi),\phi(\omega_\eta)] = \frac i{4\pi\Vert\Delta\bm x\Vert}\frac1{\sqrt{8\pi}}
			\left[e^{-\frac18(\Vert\Delta\bm x\Vert+\Delta t)^2} -
				e^{-\frac18(\Vert\Delta\bm x\Vert-\Delta t)^2}\right]\Id,
	\end{equation}
where
	\begin{equation}
		\Vert\Delta\bm x\Vert = \sqrt{\sum_{k=1}^3(\xi_k-\eta_k)^2}\qquad\text{and}\qquad
			\Delta t = \xi_0-\eta_0.
	\end{equation}
It is easily seen that the commutator falls off like a Gaussian on space-like separated
4-vectors $\xi$ and $\eta$, which means that strict locality is lost on the quantum space-time.

We close this section with a remark on field algebras for quantum fields defined
over the quantum space-time. Since there is no clear and sharp notion of regions
on quantum space-time, they can no longer be used to index the (non-local) field
algebras constructed from quantum fields on quantum space-time. We might consider
their non-commutative counterpart though, i.e. the projections in the Borel completion
of the \calg\ describing the basic model of quantum space-time. Thus for each state
$\omega\in\mathcal S(\mathcal E)$ we consider its normal extension $\tilde\omega\in
\mathcal S(\tilde{\mathcal E})$, and for every projection $E$ we define \cite{doplicher09}
\begin{subequations}
	\begin{equation}
		\A F(E) = \{W(\omega)\ |\ \tilde\omega(E) = 1\}'',
	\end{equation}
where
	\begin{equation}
		W(\omega) = e^{\frac i2[\phi(\omega)^*+\phi(\omega)^{**}]},\qquad\omega\in\mathcal S(\mathcal E)
	\end{equation}
\end{subequations}
are the Weyl operators generated by the Segal field $\phi$, and $''$ is the double commutant. A net of field algebras
is then given by the association
	\begin{equation}\label{eq:proj_net}
		E \mapsto \A F(E),\qquad E \in \tilde{\mathcal E}.
	\end{equation}
Isotony of this net is to be understood in the sense,
\begin{subequations}
	\begin{equation}
		E_1 < E_2 \qquad\Rightarrow\qquad\A F(E_1) \subset \A F(E_2),\qquad E_1,E_2\in\tilde{\mathcal E},
	\end{equation}
while covariance takes the form (see next section for the meaning of $\alpha_L$ and
$\tau_L$)
	\begin{equation}
		\alpha_L\A F(E) = \A F(\tau_L E),\qquad\forall L\in\Pup.
	\end{equation}
\end{subequations}

\subsection{The action of the Poincaré group\label{sec:action}}

Let $f$ be an element of the Banach $*$-algebra $\mathcal E_0$, and let us define,
according to \cite{doplicher95}, the following action of the full Poincaré group
$\Poin$ on $\mathcal E_0$
	\begin{equation}\label{eq:action_algebra}
		(\tau_{(\Lambda, a)} f)(\sigma,\alpha)=\det(\Lambda)e^{-i\alpha^\mu a_\mu}
			f(\Lambda^{-1}\sigma{\Lambda^{-1}}^T,\Lambda^{-1}\alpha),
	\end{equation}
where $(\Lambda,a)$ denotes the generic element of $\Poin$. Then $\tau_g$ is an isometry
with respect to the norm on $\mathcal E_0$, for each $g\in\Poin$, i.e.
	\begin{equation*}
		\Vert \tau_{(\Lambda, a)}f\Vert = \Vert f\Vert,\qquad\forall (\Lambda,a)\in\Poin,
	\end{equation*}
and moreover
	\begin{equation}
		\tau_{(\Lambda',a')} \circ \tau_{(\Lambda,a)} = \tau_{(\Lambda'\Lambda, \Lambda'a + a')},
	\end{equation}
from which it follows that $\tau:\Poin\to\Aut(\mathcal E_0)$ is a group homomorphism. Thus the above action can be extended naturally to the enveloping \calg\ $\mathcal E$
of $\mathcal E_0$, providing an action of the Poincaré group by automorphisms on the \calg\ 
of quantum space-time.

An action of the Poincaré group on the unbounded operators $\{q_\mu,\mu=0,\ldots,3\}$
is defined by transposition. Considering representations of the operators $Q_{\mu\nu}$ and $q_{\mu\nu}$
on some Hilbert space, we take the following representation $\pi$ of $\mathcal E_0$,
	\begin{equation}
		\pi(f_1\otimes f_2) = f_1(Q)\int\de^4\alpha f_2(\alpha) e^{i\alpha^\mu q_\mu},
	\end{equation}
for each $f_1\otimes f_2\in\mathcal E_0$, with $f_1\in C_0(\Sigma)$ and $f_2\in
L^1(\IR^4,\de^4\alpha)$. Under the action of \eqref{eq:action_algebra} the last relation becomes
\begin{equation}
	\pi(\tau_{(\Lambda,a)}f_1\otimes f_2) = f'_1(Q)\int\de^4\alpha\det(\Lambda)e^{-i\alpha^\mu a_\mu} f_2(\Lambda^{-1}\alpha) e^{i\alpha^\mu q_\mu},
\end{equation}
where $f_1'(\sigma)=f_1(\Lambda^{-1}\sigma{\Lambda^{-1}}^T)$. But
\begin{equation*}
	\int\de^4\alpha e^{-i\alpha^\mu a_\mu} f_2(\Lambda^{-1}\alpha) e^{i\alpha^\mu q_\mu}=\int\de^4\alpha f_2(\alpha) e^{i\alpha^\mu[\Lambda^{-1}(q - a)_\mu]},
\end{equation*}
and this allows us to define the transpose action on $q_{\mu}$ as
\begin{equation}\label{eq:tau_on_q}
	\tau^{-1}_{(\Lambda,a)}q_\mu = (\Lambda q)_\mu + a_\mu\Id.
\end{equation}

Relativistic covariance of the fields on quantum space-time is implemented in the
usual way, i.e. by defining the unitary representation of the Poincar\'e group over
the Hilbert space $\Hilb_F$ of the Fock representation. Denoting such unitary representation
by $U(L)$, $L\in\Poin$, we set $\alpha_L=\Ad U(L)$, and thus we have \cite{piacitelli11}
	\begin{equation}\label{eq:alpha_action}
		\alpha_L(a(\bv k)) = e^{i(\Lambda k)^\mu a_\mu}a(\bv k_\Lambda),
	\end{equation}
where $\bv k_\Lambda$ denotes the spatial part of $\Lambda k$.

If we now take the tensor product of both $\tau_L$ and $\alpha_L$, i.e. $\tau_L\otimes\alpha_L$,
with $L\in\Poin$, we get an action of the full Poincar\'e group $\Poin$ on
$\mathcal E\otimes\Bound(\Hilb_F)$, i.e. on the scalar field $\phi(q)$ defined in eq. \eqref{eq:scalar_field}. Hence for a generic element $L\in\Poin$ of the full Poincaré group we have
	\begin{align*}
		(\tau_L\otimes\alpha_L)\phi(q)&=\int\left[(\tau_Le^{ik^\mu q_\mu})\otimes
				(\alpha_La(\bv k))+\text{h. c.}\right]\de\Omega_m^+(\bv k)\\
			&=\int\left[e^{ik^\mu q_\mu}\otimes a(\bv k)+\text{h. c.}\right]\de\Omega_m^+(\bv k),
	\end{align*}
i.e.
	\begin{equation}
		(\tau_L\otimes\alpha_L)\phi(q) = \phi(q),
	\end{equation}
as it should be for a scalar neutral field (cf. \cite{doplicher95}). The previous equation expresses in compact form the full Poincar\'e covariance of the basic model.

	\section{The scale-covariant model\label{chap:model}}
\label{sec:scale-covariant}
This section is completely devoted to the introduction and discussion of the scale-%
covariant model of quantum space-time. This alternative model was already mentioned
in \cite{doplicher95}, and corresponds to the limiting case $\lambda_P\to0$ while keeping the non-commutativity. Scale-covariance is to be expected, for the only relevant
scale ruling the quantum structure of space-time, namely $\lambda_P$, has been removed.

As a first step we present the joint spectrum of the new central elements $R$, which plays a fundamental role in the analysis of irreducible representations. We then proceed to the construction of the \calg\ for the model
and the determination of its symmetry group.

\subsection{The quantum conditions\label{sec:sc_quantum_cond}}

Now we define the scale-covariant model as a limiting case of the basic model. In order to reveal the presence of the Planck's length, we can
switch again to generic units, thus obtaining
	\begin{equation}\label{eq:commutator_generic}
		[q_\mu, q_\nu] \subset i \lambda_P^2 Q_{\mu\nu}.
	\end{equation}
We can now perform the aforementioned limit $\lambda_P\rightarrow 0$, while keeping fixed the quantity
\begin{equation}\label{eq:new_central}
		R_{\mu\nu}:=\lambda_P^2 Q_{\mu\nu}.
	\end{equation}
The commutator \eqref{eq:commutator_generic} now reads
	\begin{equation}\label{eq:scale_cov_cr}
		[q_\mu, q_\nu] \subset iR_{\mu\nu},
	\end{equation}
and is not affected by the limit. However, the quantum conditions \eqref{eq:basic} now take the following form,
	\begin{subequations}\label{eq:scale_cov_model}
		\begin{align}
			\frac12R_{\mu\nu}R^{\mu\nu} &=0,\\
			\left(\frac12R_{\mu\nu}(*R)^{\mu\nu}\right)^2&=\lambda_P^8,\\
			[q_\lambda,R_{\mu\nu}] &= 0\qquad\forall\lambda,\mu,\nu=0,\ldots,3,
		\end{align}
	\end{subequations}
so that the second quantity is taken to vanish in the limit $\lambda_P\rightarrow 0$.

In order to determine the structure and the behavior of the joint spectrum of the central elements $R$ under the action of the Lorentz 
group we represent a point
$\sigma\in\Sigma_0$ in terms of its electric and magnetic components as discussed in the previous section. The new quantum conditions \eqref{eq:scale_cov_model} for $\lambda_P\rightarrow 0$ lead
to
	\begin{subequations}\label{eq:scale_cov_em}
		\begin{align}
			\norm{\bv m}^2 - \norm{\bv e}^2 &= 0,\\
			\bv e\cdot\bv m &= 0.
		\end{align}
	\end{subequations}
In other words $\bv m$ and $\bv e$ are now orthogonal and of equal norm in $\mathbb E^3$. 

One can notice that the classical space-time is a special case of this scale-covariant model, as the choice $\sigma=0$, i.e. $\bv e = \bv m =0$, satisfies the
above quantum conditions \eqref{eq:scale_cov_em}. Moreover this point is also a 
degenerate orbit under the action \eqref{eq:action_on_jsp} of the Lorentz group. Since we have no interest in the commutative case in this paper, in order to simplify the exposition we shall rule out this case and redefine the joint spectrum of the elements $R$ as their actual spectrum
with the point $\sigma=0$ removed.

As for the basic model, we choose a special point $\sigma_0\in\Sigma_0$, consider
the irreducible representations of the \calg\ describing the new scale-covariant 
model at $\sigma_0$, and then use it to move to any other representation at a 
different point $\sigma\in\Sigma$ by means of the action \eqref{eq:action_on_jsp}.
It turns out that, like in the basic model, $\Sigma_0$ is a single orbit under the
action of the Lorentz group, as shown by the following
	\begin{proposition} \label{prop:transitive} The action of the Lorentz group
		\eqref{eq:action_on_jsp} on the space $\Sigma_0$ is transitive.
	\end{proposition}
	\begin{proof} Let us fix an arbitrary point $\sigma_0\in\Sigma_0$, $\sigma_0=
		(\bv e_0,\bv m_0)$ and consider any other 
		point $\sigma\in\Sigma_0$, $\sigma=(\bv e,\bv m)$. Setting $\bv n_0 = \bv e_0
		\times\bv m_0$ and $\bv n = \bv e\times\bv m$, we now define the endomorphism $R
		\in\End(\IE^3)$ on $\IE^3$ given by\footnote{Since both $\{\hv e_0, \hv m_0, \hv 
		n_0\}$ and $\{\hv e, \hv m, \hv n\}$ are orthonormal basis of $\IE^3$, they 
		coincide with their own dual basis.}
			\begin{equation*}
				R := \lambda\hv e_0\otimes\hv e + \mu\hv m_0\otimes\hv m + \nu\hv n_0
					\otimes\hv n,\qquad\lambda,\mu,\nu\in\{\pm1\}.
			\end{equation*}
		It is easy to verify that $R$ maps $\{\hv e_0, \hv m_0, \hv n_0\}$ onto $\{\lambda
		\hv e, \mu\hv m, \nu\hv n\}$, i.e.
			\begin{equation*}
				R(\hv e_0) = \lambda\hv e,\qquad R(\hv m_0) = \mu\hv m,\qquad R(\hv n_0) = 
				\nu\hv n,
			\end{equation*}
		and that $RR^T = \Id_3$. Moreover
			$$\norm{R\bv x}=\norm{\bv x}, \qquad \forall \bv x\in\IE^3$$
		and therefore $R\in O(3)$. Since $\{\hv e_0, \hv m_0, \hv n_0\}$ and $\{\hv e,
		\hv m, \hv n\}$ are both orthonormal basis of $\IR^3$, then
			$$\exists\omega\in\{-1,1\}\qquad|\qquad\hv e_0\wedge\hv m_0\wedge\hv n_0 =
				\omega\hv e\wedge\hv m\wedge\hv n,$$
		whence $\det(R)=\omega\lambda\mu\nu$. If $\tilde R$ denotes the inclusion of $R$ 
		in $\Lor$ given by
			$$\tilde R = \pm1\oplus R = \begin{bmatrix}\pm1 & 0\\0 & R\end{bmatrix},$$
		then
			\begin{equation}\label{eq:rotation}
				\tilde R(\hv e_0,\hv m_0)\tilde R^T = (\pm\lambda R\hv e_0, \det(R)\mu R
					\hv m_0)\equiv(\pm\lambda\hv e,\omega\lambda\nu\hv m),
			\end{equation}
		and by setting $\lambda = \pm1$, $\nu = \pm\omega$ we get
			$$\tilde R(\hv e_0,\hv m_0)\tilde R^T \equiv(\hv e, \hv m).$$
		It remains to show that there is an element $D\in\Lor$ that is capable of 
		stretching $\bv e_0$ into $\bv e_0'$ in such a way that $\norm{\bv e_0'}=
		\norm{\bv e}$. If we take the boost $B_\beta$ along the direction of $\bv e_0
		\times\bv m_0$ and let it act on $\sigma_0$, we find
			$$B_\beta\sigma_0B_\beta^T = \sqrt{\frac{1-\beta}{1+\beta}}\,\sigma_0,$$
		i.e.
			\begin{equation}\label{eq:dilation}
				(\bv e_0,\bv m_0)\stackrel{B_\beta}\longmapsto(\lambda_\beta \bv e_0,
					\lambda_\beta\bv m_0),\qquad\lambda_\beta=\sqrt{\frac{1-\beta}{1+\beta}}.
			\end{equation}
		As $\beta$ ranges in $(-1,1)$, $\lambda$ varies in $(0,+\infty)$ and therefore 
		there must be a $\beta^*\in(-1,1)$ such that $\lambda_{\beta^*}\norm{\bv e_0}=
		\norm{\bv e}$. Setting then $D=B_{\beta^*}$ we  conclude that
			$$\forall\sigma_0,\sigma\in\Sigma_0\quad\exists \tilde R, D\in\Lor \quad|\quad
				(\tilde RD)\sigma_0(\tilde RD)^T=\sigma.$$
		Finally, since $\Lor0=\{0\}\subset\Sigma$, then $\Sigma_0=\Sigma\smallsetminus
		\{0\}$ must be stable under the action of $\Lor$. 
	\end{proof}
An immediate consequence of the above proposition is that $\Lor\sigma=\Sigma_0$ for 
any $\sigma\in\Sigma_0$ or, in other words, that $\Sigma_0$ is a single orbit under 
the action of the full Lorentz group $\Lor$.
	\begin{corollary}\label{cor:transitive} For every pair of points $\sigma_0,\sigma\in
		\Sigma_0$ there is a proper and orthochronous Lorentz transformation $\Lambda\in
		\Lup$ such that $\sigma=\Lambda\sigma_0\Lambda^T$.
	\end{corollary}
	\begin{proof} Let $R$ be as in the proof of proposition \ref{prop:transitive} and
		consider the inclusion
			$$\tilde R = 1 \oplus R.$$
		Then we must take $\lambda=1$, and consequently $\nu=\omega$, so that we remain
		with $\det(R)=\mu$. Of course we must set $\mu=1$ in order to have a proper
		rotation, i.e. $R \in SO(3)$ and so $\tilde R\in\Lup$.
	\end{proof}
	\begin{corollary} The action of $\Lor$ on $\Sigma_0$ is free.
	\end{corollary}
	\begin{proof} Let $\sigma\in\Sigma_0$ be any point and consider its representation
		in terms of electric and magnetic components, $\sigma=(\bv e, \bv m)$. Then
		$\bv e\cdot\bv m=0$ and $\norm{\bv e}=\norm{\bv m}$. The proof of proposition
		\ref{prop:transitive} shows that
			$$\forall L\in\Lor,\quad\exists \tilde R, D\in\Lor\quad | \quad L=\tilde RD,$$
		where $\tilde R$ is a rotation and $D$ a dilation. Hence there must exist $\lambda_L\in\IR^+$,
		$R_L\in O(3)$ such that $$L(\bv e, \bv m)L^T = \lambda_L(\pm R_L\bv e, \det(R_L)R_L\bv m).$$
		Let $H_\sigma$ be the stabilizer of the point $\sigma\in\Sigma_0$, and let $h\in H_\sigma$.
		Then
			$$h\sigma h^T = \sigma,$$
		which means
			$$\lambda_h(\pm R_h\bv e, \det(R_h)R_h\bv m) = (\bv e,\bv m).$$
		But the only rotation in $\IR^3$ with at least two fixed axes is the identity
		$\Id_{3}$, and therefore
			$$H_\sigma = \{\Id_\Lor\},\qquad\forall\sigma\in\Sigma_0,$$
		i.e. all the stabilizers of $\Sigma_0$ are trivial.
	\end{proof}
	\begin{corollary} \label{cor:section} Given a point $\sigma_0\in\Sigma_0$, there
		exists a continuous map
		\begin{subequations}
			\begin{align}
				\Sigma_0 &\to\Lup\\
				\sigma &\mapsto \Lambda_\sigma,
			\end{align}
		\end{subequations}
		such that $\sigma = \Lambda_\sigma\sigma_0\Lambda_\sigma^T$.
	\end{corollary}
The explicit form of the map $\sigma\mapsto\Lambda_\sigma$ is given by the proof of 
proposition \ref{prop:transitive}. Let $\sigma_0\in\Sigma_0$ be a fixed point such 
that $\sigma_0=(\bv e_0, \bv m_0)$ and $\sigma\in\Sigma_0$ a generic point, and let
$R_\sigma$ and $D_\sigma$ be the rotation $\tilde R$ and the dilation $D$ of the
proof of proposition \ref{prop:transitive} respectively. Then $\Lambda_\sigma =
R_\sigma D_\sigma$, and since both $\sigma\mapsto R_\sigma$ and $\sigma\mapsto
D_\sigma$ are continuous, so is $\sigma\mapsto\Lambda_\sigma$.
	\begin{corollary}\label{cor:identity} For any $\Lambda\in\Lor$ the following identity
			\begin{equation*}
				\Lambda\Lambda_\sigma = \Lambda_{\Lambda\sigma\Lambda^T}
			\end{equation*}
		holds.
	\end{corollary}
	\begin{proof} Since the map $\sigma\mapsto\Lambda_\sigma$ is such that $\sigma =
		\Lambda_\sigma\sigma_0\Lambda_\sigma^T$, we have
			$$\Lambda\sigma\Lambda^T = \Lambda(\Lambda_\sigma\sigma_0\Lambda_\sigma^T)\Lambda^T.$$
		On the other hand we have
			$$\Lambda\sigma\Lambda^T = \Lambda_{\Lambda\sigma\Lambda^T}\sigma_
				0\Lambda_{\Lambda\sigma\Lambda^T}^T,$$
		thus it must be $\Lambda\Lambda_\sigma = \Lambda_{\Lambda\sigma\Lambda^T}$.
	\end{proof}

\subsection{Representations of the coordinates}

The above results enable us to select a suitable special point $\sigma_0\in\Sigma_0$, e.g.
	\begin{equation}\label{eq:sigma0_mat}
		(\bv e_0,\bv m_0) = ((1,0,0),(0,-1,0))\quad\iff\quad\sigma_0=
			\begin{bmatrix}
				 0 &  1 & 0 & 0\\
				-1 &  0 & 0 & 1\\
				 0 &  0 & 0 & 0\\
				 0 & -1 & 0 & 0
			\end{bmatrix}.
	\end{equation}
As a symplectic form, $\sigma_0$ is degenerate, i.e. $\det(\sigma_0) = 0$, and as
a consequence there must be a basis of vectors in $\IR^4$ such that
	\begin{equation}\label{eq:sigma0}
		\sigma_0 \cong J \oplus 0_2,
	\end{equation}
where $J$ is the standard symplectic form on the plane $\IR^2$ given by
\eqref{eq:stdsympl}, and $0_2$ is the null $2\times2$ matrix. Since the
structure of the matrix $J\oplus 0_2$ is simpler than the original $\sigma_0$, 
it is convenient throughout the sequel to assign it the symbol
	\begin{equation}\label{eq:sstd}
		\sstd := \begin{bmatrix}
				0 & -1 & 0 & 0\\
				1 &  0 & 0 & 0\\
				0 &  0 & 0 & 0\\
				0 &  0 & 0 & 0
			\end{bmatrix},
	\end{equation}
and call it the \emph{standard} symplectic form.
It must be noted though that $\sstd$ is not a point of the joint spectrum $\Sigma_0$, 
but rather a conjugate element to each of its points. Hence we must always
\emph{lift}, in the appropriate sense, each derived result back onto the joint spectrum $\Sigma_0$. From now on we agree on employing the symbol $A$ to refer to
the non-singular matrix $A\in\Mat_{4\times4}(\IR)$ that concretely realizes the 
conjugation
	\begin{equation}\label{eq:A}
		\sstd = A\sigma_0A^T,
	\end{equation}
namely
	\begin{equation}\label{eq:explicitA}
		A=\begin{bmatrix}
				0 & 1 & 0 & 0\\
				1 &  0 & 0 & 0\\
				0 &  0 & 1 & 0\\
				1 &  0 & 0 & 1
			\end{bmatrix}.
	\end{equation}
In order to derive an irreducible representation for the coordinates $q$ corresponding to the 
point $\sigma\in\Sigma_0$ we start by considering those at the special point
$\sigma_0\in\Sigma_0$, i.e.
	\begin{equation}\label{eq:cr_sigma0}
		[q^{\sigma_0}_\mu,q^{\sigma_0}_\nu] = i(\sigma_0)_{\mu\nu}\Id.
	\end{equation}
Conjugating both sides of the above relation with the matrix $A$ we get
	\begin{equation}
		[(Aq^{\sigma_0})_\mu,(Aq^{\sigma_0})_\nu] = i\sstd_{\mu\nu}\Id,
	\end{equation}
and thus we can introduce the representation $q^{\sstd}_\mu = Aq^{\sigma_0}_\mu$, 
whose explicit form is now easier to find. For, according to the expression \eqref{eq:sstd} of $\sstd$, we must have a direct sum between a pair of Schr\"odinger operators $(p,x)$ and a pair of central operators, namely
	\begin{equation}\label{eq:qsstd}
		q^{\sstd} = \begin{bmatrix}p\\x\\a\Id\\b\Id\end{bmatrix}.
	\end{equation}
This shows that, as a consequence of the degeneracy of $\sstd$, there are two more 
parameters that label the irreducible representations of the coordinates $q$. Using 
the inverse $A^{-1}$ we can lift $q^{\sstd}$ to a covariant representation at $\sigma_0
\times(a,b)$. Explicitly we have \cite{perini}
	\begin{equation}\label{eq:qsigma0}
		q^{\sigma_0\times(a,b)} = \begin{bmatrix}x\\p\\a\Id\\b\Id-x\end{bmatrix},
	\end{equation}
and any other representation at a point $\sigma$ different from $\sigma_0$ is equivalent to the representation $\Lambda q^{\sigma_0\times(a,b)}$, i.e.
	\begin{equation}
		q^{\sigma\times(a,b)}\cong \Lambda q^{\sigma_{0}\times(a,b)},
	\end{equation}
for a suitable choice of $\Lambda$ (see corollary \ref{cor:transitive}).

Of course, on a scale-covariant model one might believe that it is possible to go beyond the Poincaré group, and consider, for instance, the scale transformations
	\begin{equation}
		q^{\sigma\times(a,b)}_\mu \mapsto\lambda q^{\sigma\times(a,b)}_\mu,\qquad
			\lambda\in(0,+\infty),
	\end{equation}
as new symmetries. Since we already have an explicit irreducible representation of
the coordinates at $\sigma_0\in\Sigma_0$, namely \eqref{eq:qsigma0}, we may consider 
the dilations on $q^{\sigma_0\times(a,b)}$ alone, for any other representation follows 
from this one, as discussed above. If we substitute $\lambda q$ for $q$ the commutation relations \eqref{eq:cr_sigma0} become
	\begin{equation}
		[\lambda q^{\sigma_0\times(a,b)}_\mu,\lambda q^{\sigma_0\times(a,b)}_\nu] =
			i(\lambda^2\sigma_0)_{\mu\nu}\Id.
	\end{equation}
This last relation requires the two representations $\lambda q^{\sigma_0\times(a,b)}$ 
and $q^{\lambda^2\sigma_0\times(a',b')}$ to be equivalent for a suitable choice of the 
parameters $(a',b')$ as functions of $(a,b)$. But the dilations on $\Sigma_0$ are 
implemented by Lorentz transformations (cf. proposition \ref{prop:transitive}) and 
therefore there must exists an element $D\in\Lup$ such that
	\begin{equation}
		q^{\lambda^2\sigma_0\times(a',b')} \cong Dq^{\sigma_0\times(a',b')},
	\end{equation}
i.e.
	\begin{equation}
		q^{\sigma_0\times(a,b)} \cong \frac1\lambda Dq^{\sigma_0\times(a',b')}.
	\end{equation}
A computation \cite{perini} shows that the new parameters $a'$ and $b'$ in terms of $a$, $b$ and $\lambda$ are given by
\begin{subequations}
	\begin{align}
		a'&=\lambda a,\\
		b'&=\lambda^3 b.
	\end{align}
\end{subequations}

\subsection{The \calg\ of the model\label{eq:scale_cov_calg}}

In order to construct the \calg\ of the scale-covariant model we may repeat the steps we have followed in the case of the basic model (see section \ref{sec:calg_basic}). 
We start with the vector space $L^1(\IR^4,\de^4\alpha)$. We then define $\mathcal E_0 = C_0(\Sigma, \A F)$ and turn it into a Banach $*$-algebra as in the case of the basic model, i.e. we introduce the multiplication, involution and norm given by, respectively
\begin{subequations}
	\begin{align}
		(f\times g)(\sigma,\alpha)&:=\int f(\sigma,\alpha')g(\sigma,\alpha-\alpha')e^{\frac i2\sigma(\alpha,\alpha')}\de^4\alpha',\\
		f^*(\sigma,\alpha)&:=\overline{f(\sigma,-\alpha)},\\
		\Vert  f\Vert &:=\sup_{\sigma\in\Sigma}\Vert f(\sigma,\ \cdot\ )\Vert_1.
	\end{align}
\end{subequations}
The multiplication $\times$ above depends on the point $\sigma\in\Sigma_0$, but the result of corollary \ref{cor:section} allows us to fix it to the special point $\sigma_0$. In fact, given the map $\Lambda_\sigma$, the element $\sigma\in\Sigma_0$, interpreted as a bilinear form on $\IR^4$, may be written as
\begin{equation}\label{eq:first_step}
	\sigma = \sigma_0\circ(\Lambda_\sigma\otimes\Lambda_\sigma).
\end{equation}
Then, for every pair of elements $f,g\in\mathcal E_0$, we have
\begin{align*}
	(f \times g)(\sigma,\alpha) &= \int f(\sigma,\alpha')g(\sigma,\alpha-\alpha')e^{\frac i2\sigma(\alpha,\alpha')}\de^4\alpha'\\
		&= \int f(\sigma,\alpha')g(\sigma,\alpha-\alpha')e^{\frac i2\sigma_0(\Lambda_\sigma\alpha,\Lambda_\sigma\alpha')}\de^4\alpha'\\
		&= \int f(\sigma,\Lambda_\sigma^{-1}\alpha')g(\sigma,\alpha-\Lambda_\sigma^{-1}\alpha')e^{\frac i2\sigma_0(\Lambda_\sigma\alpha,\alpha')}\de^4\alpha',
\end{align*}
and therefore
\begin{equation}\label{eq:fixed_mult}
	(f \times g)(\sigma,\Lambda_\sigma^{-1}\alpha)=\int f(\sigma,\Lambda_\sigma^{-1}\alpha')g(\sigma,\Lambda_\sigma^{-1}(\alpha-\alpha'))e^{\frac i2\sigma_0(\alpha,\alpha')}\de^4\alpha'.
\end{equation}
This suggests the introduction of a map $T$ between algebras given by
\begin{equation}\label{eq:Tiso}
	(Tf)(\sigma,\alpha) := f(\sigma,\Lambda_\sigma^{-1}\alpha),
\end{equation}
that is evidently one-to-one, continuous and $*$-preserving. It becomes multiplication preserving as well provided that the image algebra of $\mathcal E_0$ under $T$ is endowed with the fixed multiplication derived above (cf. equation \eqref{eq:fixed_mult}), for in this case \cite{perini}
\begin{equation}\label{eq:last_step}
	T(f\times g) = Tf\times Tg.
\end{equation}
Thus \eqref{eq:Tiso} defines a $*$-isomorphism of the algebras $\mathcal E_0$ and $\mathcal E_{\sigma_0}:=(\Ran T,\times_{\sigma_0})$. The steps from \eqref{eq:first_step} to \eqref{eq:last_step} may be repeated once more in order to replace $\sigma_0$ with $\sstd$. Indeed, we have
\begin{equation}
	\sigma_0 = \sstd\circ(A^{-1}\times A^{-1}),
\end{equation}
and thus the map
\begin{equation}
	(Sf)(\sigma,\alpha) := |\det(A)|f(\sigma,A\alpha),\qquad\forall f\in\Ran T
\end{equation}
is a new $*$-isomorphism between algebras, specifically between $(\Ran T,\times_{\sigma_0})$ and $(\Ran(S \circ T),\times_{\sstd})$. Let $\Estd$ denote the latter algebra, i.e.
\begin{equation}
	\Estd := \left(C_0(\Sigma_0, L^1(\IR^4)),\times_{\sstd}\right),
\end{equation}
then the joint spectrum of the central elements $R_{\mu\nu}$ and the two commuting coordinates of $q^{\sstd}$ may be \emph{reconstructed} by taking the Fourier transform of each element $f\in\Estd$ relatively to the two variables associated with the two commuting coordinates. By the Riemann-Lebesgue theorem then $\Fou_2 : L^1(\IR^4)\to C_0(\IR^2,L^1(\IR^2))$ is an injective $*$-homomorphism and therefore it injects the Banach $*$-algebra $\Estd$ into the Banach $*$-algebra
\begin{equation}
	\Estd^{(2)} := C_0(\Sigma_0\times\IR^2,L^1(\IR^2)).
\end{equation}

To be more definite, we shall introduce the following notation, that will occur again in the sequel. A point $x\in\IR^4$ may equally be represented as the element \begin{equation}\label{eq:qc_notation}
	x_q\oplus x_c\in \IR^2\oplus\IR^2,
\end{equation}
the subscripts $c$ and $q$ standing for \emph{classic} and \emph{quantum} respectively. The former is associated with the two commuting coordinates, while the latter to the non-commuting ones. Thus, for every $f\in\Estd$, we have
\begin{equation}
	\Fou_2[f](\sigma,x_c,\alpha_q) := \int f(\sigma,\alpha_q\oplus\alpha_c)e^{i\langle x_c,\alpha_c\rangle}\de^2\alpha_c,
\end{equation}
where $\langle\ \cdot\ ,\ \cdot\ \rangle : \IR^2\times\IR^2\to\IR$ denotes the Euclidean inner product on $\IR^2$. Now, with the above defined notation, we have (cf. \eqref{eq:sigma0} and \eqref{eq:sstd})
\begin{align*}
	\sstd(\alpha,\beta) &= \sstd(\alpha_q\oplus\alpha_c,\beta_q\oplus\beta_c)\\
		&= (J\oplus 0_2)(\alpha_q\oplus\alpha_c,\beta_q\oplus\beta_c)\\
		&= J(\alpha_q,\beta_q),
\end{align*}
and this suggests the introduction of a \emph{twisted} multiplication on $\Estd^{(2)}$ such that the map $\Fou_2$ defined above becomes multiplication preserving. From the well known properties of the Fourier transform, the multiplication of two elements $f,g\in\Estd$ may also be written as
\begin{align*}
	(f\times g)(\sigma,\alpha) &= \int f(\sigma,\alpha')g(\sigma,\alpha-\alpha')e^{\frac i2\sstd(\alpha,\alpha')}\de^4\alpha\\
		&= \int f(\sigma,\qc{\alpha'})g(\sigma,\qc{\alpha}-\qc{\alpha'})e^{\frac i2J(\alpha_q,\alpha'_q)}\de^2\alpha_q\de^2\alpha_c\\
		&= \int \Fou_2[f](\sigma,k_c,\alpha'_q)\Fou_2[g](\sigma,k_c,\alpha_q-\alpha'_q)e^{\frac i2J(\alpha_q,\alpha'_q)}e^{i\langle k_c,\alpha_c\rangle}\de^2\alpha_q\de^2k_c,
\end{align*}
and if we now apply $\Fou_2$ to the left hand side, we end up with the identity
\begin{equation}
	\Fou_2[f\times g]=\Fou_2[f]\times \Fou_2[g],\qquad\forall f,g\in\Estd,
\end{equation}
which proves that $\Fou_2$ is a $*$-homomorphism between algebras.

Since both $\Estd$ and $\Estd^{(2)}$ admit maximal \cs-seminorms that are actually norms \cite{perini}, their enveloping \calg{s} are mere completions with respect to them. Moreover $\Estd$ is dense in its enveloping \calg\ while $\Fou_2[\Estd]$ is dense in the enveloping \calg\ of $\Estd^{(2)}$, and these considerations imply the existence of a $*$-isomorphism $\psi$ such that
\begin{equation}
	\psi:\text C^*(\Estd)\to \text C^*(\Estd^{(2)})
\end{equation}
is an isometric $*$-preserving one-to-one mapping between algebras. Finally, in deriving the explicit form of the enveloping \calg\ of $\Estd^{(2)}$, we end up with the analogue of theorem \ref{eq:basic_algebra} for the scale-covariant model of quantum space-time, namely \cite{perini}
\begin{theorem} The enveloping \calg\ of the Banach $*$-algebra $\Estd^{(2)}$ with respect to its unique maximal \cs-norm is isomorphic to $C_0(\Sigma_0\times\IR^2,\mathcal K)$, where $\mathcal K$ is the \calg\ of all the compact operators on a fixed separable Hilbert space.
\end{theorem}

\subsection{The group of automorphisms}

In this section we describe how to implement the space-time symmetries on the \calg\ describing the scale-covariant model of quantum space-time constructed above.

As expected, the action of the Poincaré group $\Poin$ on the Banach $*$-algebra $\mathcal E_0$ is implemented exactly as in the case of the Banach $*$-algebra $\mathcal E_0$ describing the basic model (cf. section \ref{sec:action}), namely
\begin{equation}
	(\tau_{(\Lambda, a)} f)(\sigma,\alpha)=\det(\Lambda)e^{-i\alpha^\mu a_\mu}f(\Lambda^{-1}\sigma{\Lambda^{-1}}^T,\Lambda^{-1}\alpha),\qquad\forall f\in\mathcal E_0.
\end{equation}
Even the proof that this action defines a group of automorphisms on the enveloping \calg\ of the Banach $*$-algebra describing the scale-covariant model of quantum space-time requires only a few trivial modifications.

From now on we will use $\mathcal E_0$ to denote the Banach $*$-algebra $C_0(\Sigma_0, L^1(\IR^4,\de^4\alpha))$ describing the scale-covariant model. The symbol $\mathcal E$ will denote its completion with respect to its maximal \cs-norm, i.e. its enveloping \calg.

The scale-covariant model is characterized by the existence of a group of automorphisms on $\mathcal E_0$ which naturally extends to a group of automorphisms on $\mathcal E$ that implements the scale-covariance. The mapping
\begin{subequations}
\begin{equation}
	\Delta:\IR^+\to\End(\mathcal E_0),
\end{equation}
defined by
\begin{equation}
	(\Delta_\lambda f)(\sigma,\alpha)=\lambda^4 f(\lambda^{-2}\sigma,\lambda\alpha),\qquad f\in\mathcal E_0
\end{equation}
\end{subequations}
is a homomorphism between groups, since
\begin{align*}
	((\Delta_\lambda\circ\Delta_\mu)f)(\sigma,\alpha) &= \lambda^4(\Delta_\mu f)(\lambda^{-2}\sigma,\lambda\alpha)\\
		&= (\lambda\mu)^4f((\lambda\mu)^{-2}\sigma,(\lambda\mu)\alpha)\\
		&= (\Delta_{\lambda\mu}f)(\sigma,\alpha),\qquad\forall f\in\mathcal E_0,\;\forall\lambda,\mu\in\IR^+
\end{align*}
i.e.
\begin{equation}
	\Delta_\lambda\circ\Delta_\mu = \Delta_{\lambda\mu},\qquad\forall\lambda,\mu\in\IR^+.
\end{equation}
Moreover each map $\Delta_\lambda$, $\lambda\in\IR^+$, is a $*$-automorphism of $\mathcal E_0$ \cite{perini}, since
\begin{align*}
	\exists\Delta_\lambda^{-1} &=\Delta_{\lambda^{-1}},\qquad \forall \lambda\in\IR^+,\\
	(\Delta_\lambda f^*) &= (\Delta_\lambda f)^*,\qquad\forall f\in\mathcal E_0,\\
	\norm{\Delta_\lambda f} &= \norm f,\qquad \forall f\in\mathcal E_0,
\end{align*}
and therefore $(\Ran \Delta, \circ)$ is a group of $*$-automorphisms implementing the scale-covariance on the Banach $*$-algebra $\mathcal E_0$. Since \cite{perini}
\begin{equation}
	\norm{\Delta_\lambda f}_{\text{\cs}} \equiv \norm{f}_{\text{\cs}},\qquad\forall f\in\mathcal E_0,
\end{equation}
where $\norm{\ \cdot\ }_{\text{\cs}}$ denotes the maximal \cs-norm on $\mathcal E_0$, each $\Delta_\lambda$ extends uniquely to $\mathcal E$ and therefore the action of the group $(\Ran \Delta, \circ)$ on $\mathcal E_0$ extends to an action by $*$-automorphisms on the enveloping \calg\ $\mathcal E$.

An action of the above group of automorphisms on the coordinates $q$ may be defined as well as a transposition of the action on the elements of the \calg\ describing the scale-covariant model of quantum space-time. We again consider (cf. section \ref{sec:action}) the representations of elements of the form
\begin{equation}
	\pi(f_1\otimes f_2) = f_1(R)\int f_2(\alpha)e^{i\alpha^\mu q_\mu}\de^4\alpha,
\end{equation}
where $q$ and $R$ are being used as a shorthand for $\pi(q)$ and $\pi(R)$ respectively, and $f_1\in C_0(\Sigma_0)$, $f_2\in L^1(\IR^4,\de^4\alpha)$ are completely arbitrary. Under the action of $\Delta_\lambda$ we have
\begin{equation}
	\pi\left(\Delta_\lambda(f_1\otimes f_2)\right) =  f_1'(R)\cdot \lambda^4\int f_2(\lambda\alpha)e^{i\alpha^\mu q_\mu}\de^4\alpha,
\end{equation}
where
\begin{equation}
	f_1'(\sigma) = f_1(\lambda^{-2}\sigma),\qquad\forall\sigma\in\Sigma_0.
\end{equation}
But
\begin{equation}
	\lambda^4\int f_2(\lambda\alpha)e^{i\alpha^\mu q_\mu}\de^4\alpha = \int f_2(\alpha)e^{i\alpha^\mu(\lambda^{-1} q)_\mu}\de^4\alpha,
\end{equation}
from which we get the transpose action of $\Delta_\lambda$ on the coordinates $q$, namely
\begin{equation}\label{eq:tdilation}
	\Delta_\lambda^{-1}q = \lambda q.
\end{equation}

We now focus on the determination of the global structure of the symmetry group of the coordinates $q_\mu$, and specifically its multiplication law. Therefore we shall consider the composition between a dilation $\Delta_\lambda^{-1}$ and a Poincaré transformation $\tau_{(\Lambda,a)}^{-1}$, that, we remind, is such that (cf. section \ref{sec:action})
\begin{equation*}
	\tau_{(\Lambda,a)}^{-1}q_\mu = (\Lambda q)_\mu + a_\mu\cdot\Id.
\end{equation*}
It is easily verified that one obtains the family of automorphisms
\begin{equation}
	\tilde\tau_{(\lambda,\Lambda,a)}^{-1}:=\tau_{(\Lambda,a)}^{-1}\circ\Delta_{\lambda}^{-1},\qquad\lambda\in\IR^+,(\Lambda,a)\in\Poin,
\end{equation}
and in order to determine the composition law of the new global group, we evaluate
\begin{equation}
	(\tilde\tau_{(\lambda',\Lambda',a')}^{-1}\circ\tilde\tau_{(\lambda,\Lambda,a)}^{-1})q_\mu,
\end{equation}
which leads to
\begin{align*}
	(\tilde\tau_{(\lambda',\Lambda',a')}^{-1}\circ\tilde\tau_{(\lambda,\Lambda,a)}^{-1})q_\mu
		&= \tilde\tau_{(\lambda',\Lambda',a')}^{-1}(\lambda(\Lambda q)_\mu + a_\mu)\\
		&= \lambda'\Lambda'\left[\lambda(\Lambda q)_\mu + a_\mu\right] + a'_\mu\\
		&= (\lambda'\lambda)(\Lambda'\Lambda q)_\mu+ \lambda'(\Lambda'a)_\mu + a'_\mu,
\end{align*}
i.e.
\begin{equation}
	(\tilde\tau_{(\lambda',\Lambda',a')}^{-1}\circ\tilde\tau_{(\lambda,\Lambda,a)}^{-1}) \equiv \tilde\tau^{-1}_{(\lambda'\lambda,\Lambda'\Lambda,\lambda'\Lambda'a + a')}.
\end{equation}
The sought composition law is then
\begin{equation}
	(\lambda',\Lambda',a')(\lambda,\Lambda,a) = (\lambda'\lambda,\Lambda'\Lambda,\lambda'\Lambda'a + a'),
\end{equation}
which allows us to conclude that \cite{perini}
\begin{theorem} The scale-covariant model of quantum space-time admits the symmetry group
\begin{equation}\label{eq:semidirect}
	\mathscr G = \mathscr D \rtimes \Poin,
\end{equation}
where $\mathscr D$ denotes the dilation subgroup of $\mathscr G$.
\end{theorem}
\begin{proof}
It follows from the construction of $\mathscr G$ that $\Poin\cap\mathscr D=\{0\}$. Moreover $\mathscr D$ is stable under conjugation with each element of $\mathscr G$, i.e. it is a normal subgroup, since
\begin{equation}
	(\lambda,\Lambda,a)^{-1} = (\lambda^{-1},\Lambda^{-1},-\lambda^{-1}\Lambda^{-1}a)
\end{equation}
and
\begin{align*}
	(\lambda',\Lambda',a')(\lambda,\Id,0)(\lambda',\Lambda',a')^{-1} &= (\lambda',\Lambda',a')(\lambda{\lambda'}^{-1},{\Lambda'}^{-1},-{\lambda'}^{-1}{\Lambda'}^{-1}a')\\
		&= (\lambda,\Id,0),
\end{align*}
which proves that $\mathscr G$ is the semidirect product of $\mathscr D$ and $\Poin$.
\end{proof}

	\section{The free scale-covariant scalar neutral field\label{chap:field}}

At this point we have provided all the necessary background to proceed with the construction of a free massless scalar neutral field on the
scale-covariant model of quantum space-time. We recall that the constraint on the mass comes from the scale symmetry, for a non vanishing mass introduces an intrinsic scale, namely the Compton length
\begin{equation}
	\lambda_C = \frac \hbar{mc},
\end{equation}
that would break the scale-covariance of the field theory. On a 4-dimensional space-time we therefore expect to somehow recover the ordinary behavior
\begin{equation}\label{eq:scaling}
	\begin{cases}x\mapsto\lambda x\\\phi\mapsto\lambda^{-1}\phi\end{cases}
\end{equation}
and an explicit proof that this is indeed the case is given later in this section.

\subsection{The construction of the field}

The construction of the field is carried out as discussed in section \ref{sec:qft}
for the case of the basic model. We thus consider the invariant measure on the boundary
$\partial \bar V^+$ of the future light-cone $\bar V^+$, given by
	\begin{equation}
		\de\Omega_0^+(k)=\delta(k^2)\theta(k^0)\de^4k,
	\end{equation}
and define the massless field as
	\begin{equation}
		\phi(q) := \frac1{(2\pi)^{3/2}}\int\left[e^{ik^\mu q_\mu}\otimes a(\bv k) +
			e^{-ik^\mu q_\mu}\otimes a(\bv k)^*\right]\de\Omega_0^+(\bv k).
	\end{equation}

The map from the state space $\mathcal S(\mathcal E)$ of the \calg\ $\mathcal E$
describing the scale-covariant model of quantum space-time to operators on the Hilbert
space $\Hilb_F$ of the Fock construction is defined as in eq. \eqref{eq:field_on_state},
namely
  $$\phi(\omega) := (\omega\otimes\Id)(\phi(q)),\qquad\forall\omega\in\mathcal S(\mathcal E).$$
In order to simplify some of the computations we shall refer to the special representation
$q^{\sstd\times x_c}$, $x_c\in\IR^2$. Recalling the definition of the representation
$q^{\sigma\times x_c}$ in terms of the standard point $\sstd$, that of the matrix $A$ (see
eq. \eqref{eq:A}) and that of $q^{\sigma\times x_c}$ in terms of $q^{\sigma_0\times x_c}$,
we have
	\begin{align}
		q^{\sigma\times x_c} &\cong \Lambda_\sigma q^{\sigma_0\times x_c}\nonumber\\
			&\cong \Lambda_\sigma A^{-1}q^{\sstd\times x_c},
	\end{align}
where $\Lambda_\sigma$ is the map described in corollary \ref{cor:section}. We shall
now determine an identity that will be useful in the following computations, regarding
the Minkowskian inner product $k^\mu(q^{\sigma\times x_c})_\mu$. Employing matrix
notation for the sake of simplicity we have
	\begin{equation}\label{eq:inner_identity}
		k^\mu(q^{\sigma\times x_c})_\mu \cong (B_\sigma k)^\mu(q^{\sstd\times x_c})_\mu,
			\qquad B_\sigma = \eta {A^{-1}}^T\eta\Lambda_\sigma^{-1},
	\end{equation}
where $B_\sigma$ is such that $\det(B_\sigma)\neq0$ $\forall\sigma\in\Sigma_0$, since
$\eta=\diag(1,-1,-1,-1)$, $\Lambda_\sigma$ and $A$ are all non-singular. The explicit
expression of $\phi(\omega)$, namely
	\begin{equation}
		\phi(\omega)=\frac1{(2\pi)^{3/2}}\int\omega\left(e^{ik^\mu(q^{\sigma\times x_c})_\mu}\right)
			a(\bv k)\delta(k^2)\de^4k,
	\end{equation}
can now be manipulated by using \eqref{eq:inner_identity} in order to replace $q^{\sigma\times x_c}$
with $q^{\sstd\times x_c}$. In this way we determine
	\begin{align}\label{eq:pre_smeared}
		\phi(\omega) &= \frac1{(2\pi)^{3/2}}\int\omega\left(e^{i(B_\sigma k)^\mu(q^{
				\sstd\times x_c})_\mu}\right)a(\bv k)\delta(k^2)\de^4k\nonumber\\
			&= \frac1{(2\pi)^{3/2}}\int\omega\left(e^{ik^\mu(q^{\sstd\times x_c})_\mu}\right)
				a(\bv k_{B_\sigma^{-1}})\delta((B_\sigma^{-1}k)^2)|\det(B_\sigma^{-1})|\de^4k,
	\end{align}
where $\bv k_{B_\sigma^{-1}}$ is the spatial part of the 4-vector $B_\sigma^{-1}k$.
The states of optimal localizations can be parametrized by vectors $\xi\in\IR^4$,
but in contrast to the basic model (cf. section \ref{sec:calg_basic}) there is
just one pair of Schr\"odinger operators, for two of the four quantum coordinates
are actually \emph{classical}, i.e. they are among the generators of the center of
the algebra. This means that it suffices to consider the ground state of a one-dimensional
harmonic oscillator, whose Hamiltonian is proportional to the Euclidean norm $\norm{\bv e}$
in $\IE^3$ of the electric (or, equivalently, the magnetic) part of the point $\sigma\in\Sigma_0$.
Furthermore the value of $\norm{\bv e}$ plays the role of a \emph{scalable} fundamental
length for the scale-covariant model and as such it must show up in the uncertainty
of the measurements of the non-commuting coordinates. This follows
directly from
	\begin{equation}\label{eq:state_covariant}
		\omega_\xi\left(e^{ik^\mu(q^{\sstd\times x_c})_\mu}\right) = e^{ik^\mu\xi_\mu}
			e^{-\frac12\norm{\sstd}\norm{k_q}^2_2},
	\end{equation}
where the norm $\norm\sigma$ of a skew-symmetric bilinear form $\sigma$ is defined as
	\begin{equation}
		\norm{\sigma}^2:=\frac14\Tr(\sigma^T\sigma)=\frac12(\norm{\bv e}^2 + \norm{\bv m}^2),
	\end{equation}
and we are once again employing the notation developed towards the end of section
\ref{eq:scale_cov_calg}, equation \eqref{eq:qc_notation}, page \pageref{eq:qc_notation},
$\norm{\ \cdot\ }_2$ denoting the Euclidean norm in $\IR^2$.
The Fourier transform of the above expression \eqref{eq:state_covariant} allows us
to introduce the usual family of test functions describing localization data,
	\begin{equation}\label{eq:psi_scale_cov}
		\psi_{\omega_\xi}(x) = \delta^{(2)}(x_c-\xi_c)\cdot\frac 1{2\pi\norm{\sstd}}
			e^{-\frac1{2\norm\sstd}\norm{x_q-\xi_q}_2^2},
	\end{equation}
indexed by vectors $\xi\in\IR^4$, associated with states of optimal localization
$\omega_\xi$, $\xi\in\IR^4$.
It must be remarked though that, strictly speaking, there is no absolute meaning
of optimal localization states on the scale-covariant model of quantum space-time,
for the actual absolute minimum of the uncertainty is reached at the spectral point
$0\in\Sigma$, which is associated with the ordinary commutative space-time. To see
this it suffices to consider a state $\omega$ in the domain of the operators $[q_\mu,q_\nu]$,
for in this case the uncertainty must satisfy \cite{perini}
  $$\sum_{\mu=0}^3(\Delta_\omega q_\mu)^2 \geq \sqrt 2\int\limits_\Sigma\norm\sigma\de\mu_\omega(\sigma).$$
If $\omega$ is a pure state, then $\de\mu_\omega(\sigma) = \delta(\sigma-\sigma_\omega)\de\sigma$ for some $\sigma_\omega\in\Sigma$, and therefore
  $$\sum_{\mu=0}^3(\Delta_\omega q_\mu)^2 \geq \sqrt 2\norm{\sigma_\omega},$$
which shows that the lower bound of the uncertainty depends on points from $\Sigma$.
Thus optimal localization states on the scale-covariant model of quantum space-time
are relative to a fixed point of the joint spectrum $\Sigma_0$.

For the evaluation of the Fourier transform of the remaining terms in \eqref{eq:pre_smeared},
a new identity for the inner product $k^\mu x_\mu$ can be determined that would simplify
both calculations and notation. Switching again to matrix form, we have
	\begin{equation}
		k^\mu x_\mu = (B_\sigma^{-1}k)^\mu (\tilde B_\sigma x)_\mu,\qquad \tilde B_\sigma
			= \eta B_\sigma^T\eta,
	\end{equation}
where $\tilde B_\sigma=\Lambda_\sigma A^{-1}$ is a non-singular matrix. Hence the
aforementioned Fourier transform is given by
	$$\int a(\bv k_{B_\sigma^{-1}})\delta((B_\sigma^{-1}k)^2)e^{-ik^\mu x_\mu}|\det(B_\sigma^{-1})|\de^4k
		=\phi(\tilde B_\sigma x),$$
$\phi(\tilde B_\sigma x)$ being the usual massless scalar neutral field on classical
space-time, evaluated at the point $\tilde B_\sigma x$. Thanks to the above results,
the map $\phi(\omega)$ takes the more familiar form of a field smeared out with a
test function, namely
	\begin{equation}\label{eq:smeared}
		\phi(\omega) = \int \phi(\tilde B_\sigma x)\psi_{\omega}(x)\de^4 x.
	\end{equation}
In order to refer to the general representation of the coordinates $q_\mu$ at $\sigma\in\Sigma_0$,
we introduce the functions $\psi_{\omega}^{(\sigma)}$ associated with
$\omega(e^{ik^\mu(q^{\sigma\times x_c})_\mu})$ given by
	\begin{equation}\label{eq:psi_sigma}
		\psi^{(\sigma)}_\omega(x) := |\det(A)|\psi_\omega(\tilde B_\sigma^{-1} x).
	\end{equation}
The identity \eqref{eq:smeared} then becomes
	\begin{equation}\label{eq:smeared_cov}
		\phi(\omega) = \int \phi(x)\psi^{(\sigma)}_\omega(x)\de^4 x.
	\end{equation}
\emph{En passant} we observe that, although \eqref{eq:psi_sigma} has been given as
a definition here, this can be derived in a different way, as a consequence of what
seems to be a more natural definition, namely
	\begin{equation}
		\psi_\omega^{(\sigma)}(x) := \frac1{(2\pi)^4}\int e^{-ik^\mu x_\mu}\omega\left(e^{ik^\mu
			(q^{\sigma\times x_c})_\mu}\right)\de^4k.
	\end{equation}
Using the unitary equivalence between $q^{\sigma\times x_c}$ and $\Lambda_\sigma
A^{-1}q^{\sstd\times x_c}$, it is then easy to prove that
	\begin{equation}
		\psi^{(\sigma)}_\omega(x) = |\det(A)|\psi_\omega(A\Lambda_\sigma^{-1} x),
	\end{equation}
i.e. equation \eqref{eq:psi_sigma}.

\subsection{Covariance and locality\label{sec:cov_and_loc}}

We will now explore the properties of covariance of the scalar neutral field constructed
on the scale-covariant model of quantum space-time as discussed above. With respect
to the action of the Poincaré group, this field, say $\phi(q)$, satisfies covariance
in the same way as the scalar field constructed on the basic model of quantum space-time
(see section \ref{sec:qft}). The new feature that wasn't yet explored at this point of our argumentation is the
behavior under dilations, i.e. under the transformation laws of the type \eqref{eq:scaling}.
For a scalar field in a four dimensional space-time we expect the following ``classical''
behavior,
\begin{equation}\label{eq:scaling_dimension}
	\phi'(x') = \frac1\lambda\phi(x),\qquad x'=\lambda x,
\end{equation}
to hold for any $\lambda\in\IR^+$, which is a special case, viz. $x'(x)=\lambda x$
of the more general transformation law for a scalar field $\phi$
of scaling dimension $\Delta$ in a $d$-dimensional space-time, namely \cite{difrancesco96}
\begin{equation}
	\phi'(x') = \det(\text D_x x')^{-\Delta/d}\phi(x),
\end{equation}
where $x'\in\Diff(\IR^4)$ is a generic invertible diffeomorphism on $\IR^4$. Since the action of the group of dilations on the coordinates of quantum space-time is known (cf. eq. \eqref{eq:tdilation}), we only lack a group of automorphisms on the space of all the linear operators acting on the Hilbert space $\Hilb_F$ of the Fock construction in order to completely implement the scale covariance. Applying \eqref{eq:tdilation} to the scalar field $\phi(q)$ one has
	\begin{align*}
		(\Delta_\lambda^{-1}\otimes\Id)\phi(q) &= 
				\int\left[e^{ik^\mu (\lambda q)_\mu}\otimes a(\bv k) + \text{h.c.}\right]\de\Omega_0^+(\bv k),\\
			&= \lambda^{-2}\int\left[e^{ik^\mu q_\mu}\otimes a(\lambda^{-1}\bv k) +
				\text{h.c.}\right]\frac{\de^3\bv k}{2k_0},
	\end{align*}
and thus covariance under dilations is provided by the definition of the endomorphisms
$\delta_\lambda$ such that
	\begin{equation}\label{eq:delta}
		\delta_\lambda a(\bv k) := \lambda^{-1}a(\lambda^{-1}\bv k),\qquad\forall\bv k\in\IR^3,
	\end{equation}
for each $\lambda\in\IR^+$. But an inverse $\delta_\lambda^{-1}$ always exists
$\forall\lambda\in\IR^+$, whose explicit expression is of course given by
	\begin{equation}
		\delta_\lambda^{-1} a(\bv k) = \lambda a(\lambda\bv k),\qquad\forall\bv k\in\IR^3,\lambda\in\IR^+.
	\end{equation}
Therefore the family of maps $\delta_\lambda$, $\lambda\in\IR^+$ forms a group of
automorphisms on the set of all the creation and annihilation operators, since it
is easily seen that
\begin{equation}
	\delta_\lambda\circ\delta_\mu = \delta_{\lambda\mu}
\end{equation}
holds for any $\lambda,\mu\in\IR^+$. By the definition of the $\delta_\lambda$s we thus have the sought scale-covariance
law for the massless scalar neutral field on the scale-covariant model of quantum
space-time,
	\begin{equation}
		(\Delta_\lambda^{-1}\otimes\delta_\lambda^{-1})\phi(q) = \frac1\lambda\phi(q),\qquad
			\forall\lambda\in\IR^+.
	\end{equation}

Another property worth investigating is the locality of such a field theory, and
a massless scalar field like the one introduced above is a good candidate tool to
carry out such analysis. To this end we take two states of optimal localization,
namely $\omega_\xi$ and $\omega_\eta$, indexed by vectors $\xi,\eta\in\IR^4$, and
determine the commutator of the fields evaluated on them, viz.
	\begin{equation}
		[\phi(\omega_\xi),\phi(\omega_\eta)].
	\end{equation}
Exploiting all the results derived above and carrying out the computations we end
up with
	\begin{align*}
		[\phi(\omega_\xi),\phi(\omega_\eta)] &= i\int\Delta(\tilde B_\sigma(x-y))
				\psi_{\omega_\xi}(x)\psi_{\omega_\eta}(y)\de^4x\de^4y\\
			&= \frac{1}{(2\pi)^3}\int e^{ik^\mu(A^{-1}(\xi - \eta))_\mu}
					e^{-\norm{\sstd}\norm{({A^{-1}}^\dagger k)_q}_2^2}
					\left[\de\Omega_0^+(\bv k)-\de\Omega_0^-(\bv k)\right].
	\end{align*}
If the spatial part of $A^{-1}(\xi - \eta)$ is aligned along the direction of the
$x^2$ coordinate and if we introduce the parametrization
\begin{subequations}
	\begin{align}
			p^1 &= \norm{\bv p}\cos\phi\sin\theta,\\
			p^2 &= \norm{\bv p}\cos\theta,\\
			p^3 &= \norm{\bv p}\sin\phi\sin\theta,
	\end{align}
\end{subequations}
then, using the explicit expression of the matrix $A$ \eqref{eq:explicitA}, we determine
	\begin{align}
		[\phi(\omega_\xi),\phi(\omega_\eta)] = &\frac i{4\pi\norm{\Delta\bv x}}
				\sqrt{\frac2{8\pi\norm{\sstd}}}\cdot\nonumber\\
			\cdot&\frac1\pi\int\limits_0^\pi\de\phi\int\limits_{-1}^1\de a
			\left[\frac{\Delta t-\norm{\Delta\bv x}\cdot a}{2\norm{\sstd}F(a,\phi)^{3/2}}\right]
			\norm{\Delta\bv x} e^{-\frac{(\Delta t-\norm{\Delta\bv x}\cdot a)^2}{4\norm{\sstd}F(a,\phi)}},
	\end{align}
where $a = \cos\theta$, $F(a,\phi) = 2 - a^2 + 2\sqrt{1 - a^2}\sin\phi$, and
	\begin{equation}
		\Vert\Delta\bv x\Vert = \sqrt{\sum_{k=1}^3(A^{-1}(\xi - \eta))_k^2},\qquad
			\Delta t = (A^{-1}(\xi - \eta))_0.
	\end{equation}
The function $F(a,\phi)$ is continuous on $[-1,1]\times[0,\pi]$ and such that
$F([-1,1]\times[0,\pi]) = [1,4]$, hence there must be
a point $(a^*, \phi^*)\in[-1,1]\times[0,\pi]$ such that
	\begin{align}\label{eq:scale_cov_comm}
		[\phi(\omega_\xi),\phi(\omega_\eta)] = &2\frac i{4\pi\norm{\Delta\bv x}}
				\sqrt{\frac2{8\pi\norm{\sstd}}}\cdot\nonumber\\
			&\cdot\left[\frac{\Delta t-\norm{\Delta\bv x}\cdot a^*}{2\norm{\sstd}{F^*}^{3/2}}\right]
				\norm{\Delta\bv x} e^{-\frac{(\Delta t-\norm{\Delta\bv x}\cdot a^*)^2}{4\norm{\sstd}F^*}},
	\end{align}
with $F^* = F(a^*,\phi^*)$. The global structure of the above expression is very similar
to that of equation \eqref{eq:qcommutator}, expressing the commutator for the scalar
neutral field on the basic model of quantum space-time. Moreover the behavior along
space-like directions is also similar, falling off like a Gaussian, but this time
there is an explicit dependence on the point $\sigma\in\Sigma_0$ of the joint spectrum
of the operators $R_{\mu\nu}$ in \eqref{eq:scale_cov_cr}, i.e. on the representation of the commutation
relations among the coordinates. Since $\norm\sigma$ can be adjusted by the action
of some boosts in a special direction, namely $\hv e\times\hv m$, where $\sigma=(\bv e,\bv m)$,
the width of the Gaussian in \eqref{eq:scale_cov_comm}, which is roughly estimated by $\sqrt{\norm\sigma}$,  varies accordingly, which is in accordance with the fact that
there are no states of \emph{absolute} optimal localization, as remarked above.

\subsection{Nets of field algebras\label{sec:net}}

In the ordinary theory of Algebraic Quantum Field Theory, nets of field algebras
are usually indexed by families of regions of the usual, i.e. commutative, Minkowski
space-time. Since there is no such concept of region on a quantum space-time, a natural
generalization of this construction is obtained by considering families of projections
in the Borel completion of the \calg\ of the quantum model of space-time as index
sets for field algebras, as argued in \cite{doplicher09}.

In this paper however we have followed a different approach, for the commutativity
of two of the four coordinates allows us to define a net that can be indexed
by ordinary regions of the Minkowski space-time, although in a very special yet
\emph{non-trivial} way. In a sense, some triviality is to be expected whenever one
considers regions in commutative space-time which are too much \emph{arbitrary}, as
it will be clarified later on. Moreover scale-covariance of the net, in the usual sense of Wightman's axioms, is to be expected
to hold, and indeed this is one of the main results of this subsection.

It has already been argued that, as a consequence of the degeneracy of each point
$\sigma\in\Sigma_0$, interpreted as a symplectic form on $\IR^4$, a pair of coordinates
behaves classically, i.e. generate the center of the algebra. This is even more clear when looking at the expression of the
(generalized) function $\psi_\omega$ associated with a state $\omega$ (see eq. \eqref{eq:psi_scale_cov})
of optimal localization of the \calg\ $\mathcal E$ describing the scale-covariant
model of quantum space-time, where the presence of the $\delta^{(2)}$ distribution
gives it the \emph{look-and-feel} of a classical pure state. An obvious generalization
of \eqref{eq:psi_scale_cov} to a wider class of states $\omega\in\mathcal S(\mathcal E)$
is then given by
	\begin{equation}\label{eq:psi_struct}
		\psi_{\omega_{f,\xi}}(x) = f(x_c)\frac{1}{2\pi\norm{\sstd}}
			e^{-\frac1{2\norm{\sstd}}\norm{x_q-\xi}_2^2},
	\end{equation}
for any $\xi\in\IR^2$ and positive function $f\in\mathscr D(D)$, $D\subset\IR^2$ being
a bounded subset. This structure motivates the introduction of the following notation
\begin{subequations}\label{eq:states}
	\begin{align}
		&\tilde\Omega = \{\omega_{f,\xi}\in\mathcal S(\mathcal E)\ |\ \xi\in\IR^2,
			f\in\D(D),f>0,D\subset\IR^2\text{ bounded}\}\\
		&\Omega(\{x_q\}\times D) = \{\omega_{f,\xi}\in\tilde\Omega\ |\ \xi = x_q,\supp f\subset D\},
	\end{align}
\end{subequations}
with $x_q\in\IR^2$ and $D\subset\IR^2$ bounded. We shall also define and use the
following abuse of notation
	\begin{equation}\label{eq:supp}
		\supp \psi_{\omega_{f,\xi}} := \{\xi_q\}\times\supp f,\qquad\omega_{f,\xi}\in \tilde\Omega,
	\end{equation}
and it'll prove convenient to associate triples $(\sigma,x_q,D)$, consisting of a point
$\sigma$ of the joint spectrum $\Sigma_0$, an element $x_q\in\IR^2$ and a bounded
region $D\subset\IR^2$, with those special regions in $\IR^4$ of the form
	\begin{equation}
		(\sigma, x_q,D) := \Lambda_\sigma A^{-1}(\{x_q\}\times D)\subset\IR^4
	\end{equation}
of classical Minkowski space-time. We shall denote the set of all these regions $(\sigma,x_q,D)$
with $\mathcal R$ or, in other words, we define
	\begin{equation}\label{eq:regions}
		\mathcal R := \{(\sigma,x_q,D)\ |\ \sigma\in\Sigma_0,x_q\in\IR^2,D\subset\IR^2\text{ bounded}\}.
	\end{equation}
The above set can be given a preordered set structure by endowing it with the following
preorder relation. If $R=(\sigma,x_q,D)$ and $R'=(\sigma',x'_q,D')$ are any two regions
in $\mathcal R$, then
	\begin{equation}\label{eq:preorder}
		R \subset R' \qquad\iff\qquad x_q \equiv x'_q\quad\wedge\quad D \subset D'.
	\end{equation}
It is easy to see that
	\begin{proposition} The family of regions $\mathcal R$ endowed with the preorder
		relation \eqref{eq:preorder} is a directed set.
	\end{proposition}
We shall also consider the properties of $\mathcal R$ under the action of the symmetry
group \eqref{eq:semidirect}. Towards the end of this section we will make use of
the following
	\begin{lemma}\label{lem:covariance} The family of regions of classical space-time
		$\mathcal R$ is stable under the action of the symmetry group $\mathscr G$ of
		the scale-covariant model of quantum space-time given by \eqref{eq:semidirect}.
	\end{lemma}
	\begin{proof} Since the group $\mathscr G$ is the semidirect product of three factors,
	we can consider the action of each separately.

	\textit{Lorentz transformations.} For a pure Lorentz transformation $\Lambda\in\Lup$
	and a region $(\sigma,x_q,D)\in\mathcal R$ we have
		$$\Lambda(\sigma,x_q,D) = (\Lambda\sigma\Lambda^T,x_q,D)$$
	where use has been made of corollary \ref{cor:identity}. Since $\Sigma_0$ is stable
	under $\Lup$, it follows that $(\Lambda\sigma\Lambda^T,x_q,D)\in\mathcal R$, and
	this proves stability of $\mathcal R$ under $\Lup$.

	\textit{Translations.} For a translation $T_a(\sigma, x_q,D)=(\sigma, x_q,D) + a$,
	$a\in\IR^4$ we observe that $\Lambda_\sigma A^{-1}$ is invertible, and so
		$$(\sigma,x_q,D) + a = \Lambda_\sigma A^{-1}\left(\{x_q + (A\Lambda_\sigma^{-1}a)_q\}
			\times(D + (A\Lambda_\sigma^{-1}a)_c)\right),$$
	i.e. the stability of $\mathcal R$ under translations.

	\textit{Dilations.} Finally, for a dilation $\Delta_\lambda(\sigma,x_q,D)=\lambda(\sigma,x_q,D)$,
	$\lambda\in\IR^+$, we have
		\begin{equation*}
			\lambda (\sigma,x_q,D) = (\sigma,\lambda x_q,\lambda D),
		\end{equation*}
	which ultimately shows the stability of $\mathcal R$ under $\mathscr G$.
	\end{proof}

It is clear from \eqref{eq:states} that states of optimal localization are independent
from points of the joint spectrum $\sigma\in\Sigma_0$, for one of the roles played
$\Sigma_0$ is to parametrize the irreducible representations of the commutation relations
among the coordinates $q$. One then expects the field algebras associated with the
regions $(\sigma,x_q,D)\in\mathcal R$ to be equivalent for each $\sigma\in\Sigma_0$.
As a consequence of this consideration we may take the bicommutant of the algebra
over $\IC$ generated by the family of Weyl's operators
	\begin{equation}
		W(\omega) = e^{\frac i2 [\phi(\omega)^* + \phi(\omega)^{**}]},\qquad \omega\in\Omega(\{x_q\}\times D),
	\end{equation}
where $\phi(\omega)^{**}\equiv\overline{\phi(\omega)}$ is the closure of $\phi(\omega)$, and denote it with
	\begin{equation}
		\A F(\{x_q\}\times D) := \{W(\omega)\ |\ \omega\in\Omega(\{x_q\}\times D)\}''.
	\end{equation}
We then have a net of field algebras indexed by regions from $\mathcal R$ by simply
defining
\begin{subequations}\label{eq:special_net}
	\begin{equation}
		(\sigma,x_q,D) \mapsto \A F(\{x_q\}\times D),\qquad(\sigma,x_q,D)\in\mathcal R,
	\end{equation}
or, more formally,
	\begin{equation}
		R \mapsto \A F\left(\pi(R)\right), \qquad R\in\mathcal R
	\end{equation}
where $\pi:\mathcal R \to 2^{\IR^4}$ is defined by
	\begin{equation}
		\pi((\sigma,x_q,D)) = \{x_q\} \times D.
	\end{equation}
\end{subequations}

We shall now discuss three main aspects of the net constructed above, namely \emph{isotony},
\emph{locality} and \emph{covariance}. The first property is satisfied by construction,
for it is obvious that
	\begin{equation}
		R_1\subset R_2\quad\Rightarrow\quad \A F(\pi(R_1))\subset\A F(\pi(R_2)),\qquad R_1,R_2\in\mathcal R.
	\end{equation}
Locality is inevitably lost, as shown in the previous sections, for the commutator
of the fields falls off like a Gaussian on space-time separations (cf. eq. \eqref{eq:scale_cov_comm}).
It remains to investigate the property of covariance of the net \eqref{eq:special_net}.
To this end we note that the scale-covariant field $\phi$ of the previous section
satisfies a law of covariance that can be considered as the analogue of Wightman's
second axiom for quantum field theory \cite{wightman64}. Let $L\in\Poin^\uparrow_+$,
$L=(\Lambda,a)$, be any proper Poincaré transformation, and let $\alpha_L$ be the
adjoint action of $\Poin^\uparrow_+$ on the operators on the Hilbert space $\Hilb_F$
of the Fock construction (cf. eq. \eqref{eq:alpha_action}). If $\omega\in\mathcal S(\mathcal E)$
is any state of the \calg\ describing quantum space-time, then
	\begin{align*}
		\alpha_L\phi(\omega) &= \frac1{(2\pi)^{3/2}}\int\left[\omega(e^{ik^\mu q_\mu})\alpha_La(\bv k) +
				\text{h.c.}\right]\de\Omega_m^+(\bv k)\\
			&= \frac1{(2\pi)^{3/2}}\int\left[\omega(e^{ik^\mu (\Lambda q + a\Id)_\mu})a(\bv k) +
				\text{h.c.}\right]\de\Omega_m^+(\bv k),
	\end{align*}
and moreover, by defining the action $\tau_L$ on states by means of
	\begin{equation}
		(\tau_L\omega)(e^{ik^\mu q_\mu}):=\omega(e^{ik^\mu(\tau_Lq)_\mu})
	\end{equation}
we may set
	\begin{equation}\label{eq:alpha_omega}
		\alpha_L\phi(\omega)=\phi(\tau_L^{-1}\omega),
	\end{equation}
or using the shorthand notation $\omega_L:=\tau_L\omega$, also
	\begin{equation}
		\alpha_L\phi(\omega)=\phi(\omega_{L^{-1}}).
	\end{equation}
As claimed earlier, this relation looks like, and so can be regarded as, the analogue
of Wightman's second axiom for quantum field theory \cite{wightman64} for the case
of a scalar neutral field. For a state $\omega_\xi$ of optimal localization, indexed
by $\xi\in\IR^4$, the above relation becomes
	\begin{align*}
		(\tau_L^{-1}\omega_\xi)(e^{ik^\mu q_\mu}) &= \omega_\xi(e^{ik^\mu(\Lambda q+ a\Id)_\mu})\\
			&= \omega_{L\xi}(e^{ik^\mu q_\mu}),\qquad\forall L\in\Pup,
	\end{align*}
i.e.
	\begin{equation}
		\tau_L^{-1}\omega_\xi = \omega_{L\xi},\qquad\forall L\in\Pup.
	\end{equation}
The above result allows us to transfer the action of the proper Poincaré group $\Pup$
from states $\omega\in\mathcal S(\mathcal E)$ to the associated functions $\psi_\omega(x)$
defined by \eqref{eq:psi}, and ultimately to their ``support'' (see \eqref{eq:supp}).
Using \eqref{eq:psi} we determine
	\begin{align*}
		\psi_{\omega_{L^{-1}}}(x) &= \frac1{(2\pi)^4}\int e^{-ik^\mu x_\mu}(\tau_L^{-1}\omega)
				(e^{ik^\mu q_\mu})\de^4k\\
			&= \frac1{(2\pi)^4}\int e^{-ik^\mu (\Lambda^{-1}(x-a))_\mu}\omega(e^{ik^\mu q_\mu})\de^4k,
	\end{align*}
i.e.
	\begin{equation}
		\psi_{\omega_{L^{-1}}}(x) = \psi_\omega(L^{-1}x),\qquad\forall L\in\Pup,
	\end{equation}
and thus we set
	\begin{equation}\label{eq:poin_psi}
		(\tau_L\psi_\omega)(x) := \psi_\omega(Lx),\qquad\forall L\in\Pup.
	\end{equation}
Exploiting the equivalent form \eqref{eq:smeared_cov} for $\phi(\omega)$, we may
give yet another form for the covariance of the field $\phi$. Using \eqref{eq:psi_sigma}
and the result of corollary \ref{cor:identity}, we determine
	\begin{align*}
		(\alpha_L\phi)(\omega) &= \phi(\omega_{L^{-1}})\\
			&= \int\phi(Lx)\psi_\omega^{(\sigma)}(x)\de^4x,\qquad\forall L\in\Pup,
	\end{align*}
and hence
	\begin{equation}
		(\alpha_L\phi)(x) = \phi(Lx),\qquad\forall L\in\Pup.
	\end{equation}

Restricting now our attention to states $\omega$ from $\Omega(\{\xi_q\}\times D)$, for some $\xi_q\in\IR^2$ and $D\subset\IR^2$ bounded, we can propagate the action of the Poincaré group to subsets $R\in\mathcal R$ of $\IR^4$, i.e. those used to index the net of field algebras \eqref{eq:special_net}. With the abuse of notation \eqref{eq:supp}, we see from \eqref{eq:psi_sigma} that the relation
	\begin{equation}
		\supp\psi_\omega^{(\sigma)} = \Lambda_\sigma A^{-1}\supp\psi_\omega
	\end{equation}
between supports must hold, and this proves the covariance for the net under the action of the Poincaré group in the usual sense, for if we now consider $\Lambda\sigma\Lambda^T$, then
	\begin{align}
		\supp\psi_\omega^{(\Lambda\sigma\Lambda^T)} &= \Lambda_{\Lambda\sigma\Lambda^T}
				A^{-1}\supp\psi_\omega\nonumber\\
			&= \Lambda\supp\psi_\omega^{(\sigma)}.
	\end{align}
But according to definition \eqref{eq:psi_sigma}, a generic Poincaré transformation
$L\in\Pup$, $L=(\Lambda,a)$ is such that
	\begin{equation}
		\psi^{(\sigma)}_\omega(L^{-1}x) = \psi^{(\Lambda\sigma\Lambda^T)}_\omega(x-a),
	\end{equation}
and thus
	\begin{equation}\label{eq:supp_poincare}
		\supp\psi_{\tau_L^{-1}\omega}^{(\sigma)} = \Lambda\supp\psi_\omega^{(\sigma)} + a.
	\end{equation}

By repeating all the above steps for the action of the group of dilations we can
prove covariance of the net \eqref{eq:special_net} under the action of the full group
of automorphisms \eqref{eq:semidirect}. In fact, recalling the definition \eqref{eq:delta},
one has
	\begin{align*}
		\delta_\lambda\phi(\omega) &= \frac1{(2\pi)^{3/2}}\int\left[\omega(e^{ik^\mu q_\mu})
				\delta_\lambda a(\bv k)+\text{h.c.}\right]\de\Omega_0(\bv k)\\
			&= \frac\lambda{(2\pi)^{3/2}}\int\left[\omega(e^{ik^\mu (\lambda q)_\mu})
				a(\bv k)+\text{h.c.}\right]\de\Omega_0(\bv k),
	\end{align*}
and hence, for any $\lambda\in\IR^+$,
	\begin{equation}\label{eq:delta_omega}
		\delta_\lambda\phi(\omega) = \lambda\phi(\Delta^{-1}_\lambda\omega),
	\end{equation}
where $\Delta_\lambda^{-1}\omega$ is such that
	\begin{equation}
		(\Delta_\lambda\omega)(e^{ik^\mu q_\mu}):= \omega(e^{ik^\mu(\Delta_\lambda q)_\mu}),
			\qquad\lambda\in\IR^+.
	\end{equation}
For a state $\omega_\xi$, $\xi\in\IR^4$ of optimal localization, the previous identity \eqref{eq:delta_omega} implies the covariance in the form
	\begin{align*}
		(\Delta_\lambda^{-1}\omega_\xi)(e^{ik^\mu q_\mu}) &= \omega_\xi(e^{ik^\mu(\Delta_\lambda^{-1}q)_\mu})\\
			&= \omega_\xi(e^{ik^\mu(\lambda q)_\mu})\\
			&= \omega_{\lambda\xi}(e^{ik^\mu q_\mu}),
	\end{align*}
i.e.
	\begin{equation}
		\Delta_\lambda^{-1}\omega_\xi = \omega_{\lambda\xi},\qquad\forall \lambda\in\IR^+.
	\end{equation}
The action of a dilation on the function $\psi_\omega(x)$ associated with the state $\omega\in\mathcal S(\mathcal E)$ is, as before, a consequence of its definition \eqref{eq:psi}, viz.
	\begin{align*}
		\psi_{\Delta_\lambda^{-1}\omega}(x) &= \frac1{(2\pi)^4}\int e^{-ik^\mu x_\mu}
				(\Delta_\lambda^{-1}\omega)(e^{ik^\mu q_\mu})\de^4k\\
			&= \lambda^{-4}\psi_\omega(\lambda^{-1}x),
	\end{align*}
and thus we may define the action
	\begin{equation}\label{eq:delta_psi}
		(\Delta_\lambda\psi_\omega)(x) = \lambda^4\psi_\omega(\lambda x),\qquad\forall\lambda\in\IR^+.
	\end{equation}
In order to retrieve the analogue of \eqref{eq:scaling_dimension} for the scalar
field on quantum space-time we recall that
	\begin{equation*}
		\phi(\omega) = \int\phi(x)\psi_\omega(x)\de^4x,
	\end{equation*}
whence
	\begin{align*}
		\delta_\lambda\phi(\omega) &= \lambda\phi(\Delta_\lambda^{-1}\omega)\\
			&= \lambda\int\phi(\lambda x)\psi_\omega(x)\de^4x,
	\end{align*}
which suggests the introduction of the family of maps
	\begin{equation}
		(\delta_\lambda\phi)(x)=\lambda\phi(\lambda x),\qquad\lambda\in\IR^+.
	\end{equation}
Setting $\phi'=\delta_\lambda^{-1}\phi$ and $x'=\lambda x$, the last relation implies
	\begin{equation}
		\phi'(x')=\lambda^{-1}\phi(x),
	\end{equation}
which has the same structure of \eqref{eq:scaling_dimension}.

Considering now the functions \eqref{eq:psi_sigma} associated with states $\omega$
acting on the representation of the commutation relation among the coordinates at
$\sigma\in\Sigma_0$, and applying \eqref{eq:delta_psi}, we finally determine
	\begin{equation}
		\psi^{(\sigma)}_{\Delta^{-1}\omega}(x) = \lambda^{-4}\psi_\omega^{(\sigma)}(\lambda^{-1}x),
	\end{equation}
which implies covariance for the support \eqref{eq:supp} in the usual sense, i.e.
	\begin{equation}\label{eq:supp_dilations}
		\supp\psi^{(\sigma)}_{\Delta^{-1}_\lambda\omega}=\lambda\supp\psi^{(\sigma)}_\omega.
	\end{equation}

The last step towards the proof of covariance of the net \eqref{eq:special_net} under
the action of the full group \eqref{eq:semidirect} is performed by considering the
composition between a dilation and a Poincaré transformation. Taking any element
$(\lambda,\Lambda,a)\in\mathscr G$, one gets
	\begin{equation}\label{eq:supp_full_group}
		\supp\psi^{(\sigma)}_{(\tau^{-1}_L\circ\Delta^{-1}_\lambda)\omega}=
			\lambda\Lambda\supp\psi^{(\sigma)}_\omega + a.
	\end{equation}

In order to summarize the former considerations in this section, we state and
prove the following
	\begin{theorem}\label{thm:net_prop} The net of field algebras \eqref{eq:special_net}
		has the following properties.
			\begin{enumerate}[i.]
				\item \parbox{2.5cm}{isotony:} $R_1\subset R_2\quad\Rightarrow\quad
					\A F(\pi(R_1))\subset \A F(\pi(R_2)),\quad R_1,R_2\in\mathcal R$;
				\item \parbox{2.5cm}{covariance:} $g\A F(\pi(R))=\A F(\pi(gR)),\qquad
					\forall R\in\mathcal R,g\in\mathscr G$.
			\end{enumerate}
	\end{theorem}
	\begin{proof} The first property is obvious and was discussed earlier, therefore
		it remains to show that covariance holds. To this end we take a state
		$\omega\in\Omega(\pi(R))$, $R\in\mathcal R$ and a $g\in\mathscr G$,
		$g=(\lambda,\Lambda,a)$, and so we get (see eqs \eqref{eq:alpha_omega} and
		\eqref{eq:delta_omega})
			\begin{equation*}
				(g\phi)(\omega) = \lambda\phi(g^{-1}\omega).
			\end{equation*}
		Considering the functions $\psi_\omega$ and $\psi_{g^{-1}\omega}$ associated
		with $\omega$ and $g^{-1}\omega$ respectively we also have (see eq. \eqref{eq:supp_full_group}
		and lemma \ref{lem:covariance})
			\begin{equation*}
				\supp\psi_{g^{-1}\omega} = g\supp\psi_\omega,
			\end{equation*}
		whence $g^{-1}\omega\in\Omega(\pi(gR))$, and therefore
			\begin{equation*}
				g\A F(\pi(R)) \subset \A F(\pi(gR)).
			\end{equation*}
		But $\A F(\pi(gR))=gg^{-1}\A F(\pi(gR))$, and thus
			\begin{equation*}
				\A F(\pi(gR)) \subset g\A F(\pi(R)),
			\end{equation*}
		which proves the assertion.
	\end{proof}

We now turn our attention to the properties of the vacuum state $\Omega\in\Hilb_F$
relatively to the net \eqref{eq:special_net}. The embedding \eqref{eq:embedding}
allows us to define a family of real Hilbert subspaces of $\Hilb$, namely
	\begin{equation}
		K(\{x_q\}\times D) := \{E\psi_\omega\ |\ \omega\in\Omega(\{x_q\}\times D)\},
	\end{equation}
for any $x_q\in\IR^2$ and $D\subset\IR^2$ bounded. We may use the set of regions
$\mathcal R$ to index this family of subspaces, for we may introduce the mapping
	\begin{equation}
		R \mapsto K(\pi(R)),\qquad R\in\mathcal R,
	\end{equation}
the map $\pi$ being that defined in \eqref{eq:special_net}.
	\begin{proposition}\label{prop:sep_not_cyc}The vacuum vector $\Omega\in\Hilb_F$
		is separating but not cyclic for $\A F(\pi(R))$, for any $R\in\mathcal R$.
	\end{proposition}
	\begin{proof} Thanks to lemma \eqref{lem:R_of_K} it suffices to show that $K(\pi(R))$
		is separating, $R\in\mathcal R$, in order to prove that $\Omega\in\Hilb_F$ is
		separating for $\A F(\pi(R))$. But this property follows immediately from the
		definition of $K(\pi(R))$, for any $R\in\mathcal R$.

		To prove that $\Omega\in\Hilb_F$ is not cyclic for $\A F(\pi(R))$ for any $R\in\mathcal R$,
		we recall that the functions $\psi_\omega$ associated with states $\omega\in\Omega(\pi(R))$,
		$R\in\mathcal R$ have the structure \eqref{eq:psi_struct}, i.e. the product of
		a function in $\D(D)$ and a Gaussian with fixed \emph{center}. Thus it is possible
		to choose an element $f\in\Hilb$ satisfying
			\begin{equation*}
				2f(p) = f(p)+\overline{f(-p)}
			\end{equation*}
		that is not the limit of any sequence of elements in $K(\pi(R))$, for any $R\in\mathcal R$.
	\end{proof}

The above result applies to the field algebra $\A F(\{x_q\}\times\IR^2)$ for any 
fixed $x_q\in\IR^2$ as well. But if fails whenever one considers the field algebra
generated by the union of $\A F(\{x_q\}\times D)$, with $x_q$ varying inside some
open (and possibly bounded) subset of $\IR^2$, say a neighborhood $D_q\in\mathcal N(\xi_q)$
of some point $\xi_q\in\IR^2$ of the real plane, the region $D\subset\IR^2$ being
kept fixed. This claim can be stated more precisely in the form of a theorem, which
is another of the main results of this paper, but before doing so we state and prove
three lemmas that will simplify the main proof.
	\begin{lemma}\label{lem:completeness} Let $g_\xi\in\Sch(\IR^2,\IR)$ be the Gaussian
		function centered at $\xi\in\IR^2$, i.e. $g_\xi=(2\pi)^{-1}e^{-\frac12\norm{x-\xi}_2^2}$,
		$D_q\in\Neigh(\xi_0)$ a neighborhood of $\xi_0\in\IR^2$, $\mathcal G(D_q)\subset\Sch(\IR^2,\IR)$
		the set of functions $\mathcal G(D_q):=\{g_\xi\ |\ \xi\in D_q\}$ and $\Sch(\IR^2\times D_c,\IR)
		:= \{f\in\Sch(\IR^4,\IR)\ |\ \supp f\subset \IR^2\times D_c\}$ for any non-empty
		open subset $D_c\subset\IR^2$. Then
			\begin{equation}
				\Sch(\IR^2\times D_c,\IR) \subset \overline{\Span_{\IR}\mathcal G(D_q)\otimes \D(D_c)}.
			\end{equation}
	\end{lemma}
	\begin{proof} For simplicity let us assume that $\xi_0=0$. For any $f\in\D(D_c)$
		the set $\Span_{\IR}(\mathcal G(D_q)\otimes\{f\})$ is a vector subspace of
		$\Span_{\IR}\mathcal G(D_q)\otimes \D(D_c)$. Each translation $\xi\in D_q\smallsetminus\{0\}$
		of the Gaussian function $g_\xi$ may be expanded in a MacLaurin series, viz.
			\begin{equation*}
				g_\xi(x) = \sum_{k=0}^{\infty}\frac1{k!}\sum_{I\in\{1,2\}^k}
					\left(\pder{}{\eta^I}\At_{0}g_\eta(x)\right)\xi^I,
			\end{equation*}
	where $I$ is a multi-index. Introducing the family of multi-indices $\mathcal I
	:= \ds\bigcup_{k=0}^\infty\{1,2\}^k$ we thus have that $\Span_{\IR}\left(\mathcal G(D_q)
	\otimes\{f\}\right)$ contains all the functions of the form $s^Ig_0\otimes f$,
	$I\in\mathcal I$, where $s^i$ is the operator of multiplication by $x^i$, i.e. $(s^if)(x)=x^if(x)$.
	Now the set $\{h_I\ |\ h_I = s^Ig_0\}_{I\in\mathcal I}$ is total in $L^2(\IR^2)$,
	for if we take $\phi\in L^2(\IR^2)$ such that $(h_I,\phi)=0$ for all $I\in\mathcal I$,
	then the function
		\begin{align*}
			\Phi(z) &= (e^{\langle z,s\rangle}g_0,\phi)\\
				&= \sum_{k=0}^\infty\frac1{k!}\sum_{I\in\{1,2\}^k}(s^Ig_0,\phi)\bar z^I\\
				&= \sum_{k=0}^\infty\frac1{k!}\sum_{I\in\{1,2\}^k}(h_I,\phi)\bar z^I,
		\end{align*}
	would be equal to 0 everywhere on $\IC$. But for $z=iy$, with $y\in\IR^2$, $\Phi$
	becomes the Fourier transform of $\phi g_0$, and since $g_0$ is non vanishing it
	must be $\phi=0$. Thus for each $f\in\D(D_c)$ the vector space $\Span_{\IR}\mathcal G(D_q)
	\otimes \D(D_c)$ contains a total set $\{h_I\otimes f\}_{I\in\mathcal I}$ for
	$L^2(\IR^2)\otimes\{f\}$.

	Given $\psi\in\Sch(\IR^2\times D_c)$ we may consider the family of functions
	$\{\psi_I(x_c)\}_{I\in\mathcal I}$ given by
		\begin{equation*}
			\psi_I(x_c) := \int h_I(x_q)\psi(x_q,x_c)\de^2 x_q,\qquad I\in\mathcal I,
		\end{equation*}
	that are obviously such that $\{\psi_I(x_c)\}_{I\in\mathcal I}\subset\D(D_c)$.
	Besides, the previous discussion on completeness implies
		\begin{equation}
			\psi(x_q,x_c) = \sum_{I\in\mathcal I} \psi_I(x_c)h_I(x_q),\qquad\forall x_q
				\oplus x_c\in\IR^2\times D_c,
		\end{equation}
	and this shows that $\psi\in\overline{\Span_{\IR}\mathcal G(D_q)\otimes \D(D_c)}$.
	\end{proof}
	\begin{lemma}\label{lem:trivial_subspace} Let $D^{(3)}\subset\IR^3$ be a non-empty
		simply connected and bounded open subset of $\IR^3$ with regular boundary, and
		let $\mathcal N$ be the set of \emph{nice} regions in $\IR\times D^{(3)}$, i.e.
		the set of all the regions $O\subset \IR\times D^{(3)}$ satisfying the hypotheses
		of theorem \ref{thm:araki}. Then
			$$\bigcap_{O\in \mathcal N}O'=\varnothing.$$
	\end{lemma}
	\begin{proof} Suppose that $\bigcap_{O\in \mathcal N}O'\neq\varnothing$. Then
		$\bigcup_{O\in \mathcal N}O$ must be bounded, for if this were not the case,
		then $\left(\bigcup_{O\in \mathcal N}O\right)'=\varnothing$. Hence there must
		exist a nice region $R\in \mathcal N$ such that $\bigcup_{O\in \mathcal N}O\subsetneq R$,
		which is absurd.
	\end{proof}
	In \ref{app:vn} we have recalled the definition of a net of local von
	Neumann algebras indexed by regions of $\IR^4$. The net of local field algebras
	$O\mapsto\A R(O)$, $O\subset\IR^4$ for the free scalar neutral field is defined
	in terms of the net of real Hilbert subspaces given by \cite{araki64}
		\begin{equation}
			O\mapsto K(O) := \overline{E\D(O,\IR)}.
		\end{equation}
	Thus we set (cf. \eqref{eq:local_field_algebra})
		\begin{equation}
			\A R(O) := \A R(K(O)).
		\end{equation}

	\begin{lemma}\label{lem:time_irred} Let $D^{(3)}$ be as in the previous lemma.
		The field algebra $\A R(\IR\times D^{(3)})$ is irreducible.
	\end{lemma}
	\begin{proof} Let $\mathcal N$ be the set of \emph{nice} regions defined as in
		lemma \ref{lem:trivial_subspace}. For any $O\in \mathcal N$ isotony implies
			$$\A R(O)\subset \A R(\IR\times D^{(3)}),\qquad\forall O\in \mathcal N,$$
		and passing to the commutant one obtains
			$$\A R(\IR\times D^{(3)})'\subset \A R(O)',\qquad\forall O\in \mathcal N.$$
		Since $\A R(O) = \A R(K(O))$ and $K(O)$ is standard for any $O\in \mathcal N$,
		then lemma \ref{lem:araki_haag_duality} implies that $\A R(K(O))' = \A R(K(O)')
		= \A R(K(O'))$, i.e. $\A R(O)' = \A R(O')$. Moreover, for any region $O\subset\IR^4$
		the relation $\A R(O')\subset\A R(O)'$ holds, and therefore
			$$\A R((\IR\times D^{(3)})') \subset \A R(\IR\times D^{(3)})'\subset \A R(O'),
				\qquad\forall O\in \mathcal N,$$
		or equivalently
			$$\A R((\IR\times D^{(3)})') \subset \A R(\IR\times D^{(3)})'\subset
				\bigcap_{O\in \mathcal N}\A R(O'),$$
		and since $\A R(\varnothing) = \{\Id\}$ then lemma \ref{lem:trivial_subspace}
		implies
			$$\{\Id\}\subset\A R(\IR\times D^{(3)})'\subset \{\Id\}.$$
		Hence $\A R(\IR\times D^{(3)})$ is irreducible.
	\end{proof}

We are now ready to prove the following
	\begin{theorem}\label{thm:irred} Let $D_q$ be any neighborhood of a fixed point
		$x_q\in\IR^2$ and let $D_c$ be a non-empty simply connected and bounded open
		subset of $\IR^2$ with regular boundary. The field algebra $\A F(D_q\times D_c)$
		generated by
			$$\A F(D_q\times D_c) := \bigvee_{x_q\in D_q}\A F(\{x_q\}\times D_c)$$
		is irreducible.
	\end{theorem}
	\begin{proof} It is clear from the definition and lemma \ref{lem:completeness}
		that $\A F(D_q\times D_c)$ contains the field algebra $\A R$ associated with
		the real Hilbert subspace $\overline{E\Sch(\IR^2\times D_c,\IR)}$, where
		$\Sch(\IR^2\times D_c,\IR)$ is the set of functions defined in lemma \ref{lem:completeness}.
		Let $D^{(3)}$ be as in lemma \ref{lem:trivial_subspace}, but such that $D^{(3)}
		\subset\IR\times D_c$. Then the property of isotony requires that
			$$\A R(\IR\times D^{(3)}) \subset \A F(D_q\times D_c).$$
		Taking the commutant of both sides we obtain
			$$\A F(D_q\times D_c)'\subset \A R(\IR\times D^{(3)})',$$
		and since $\A R(\IR\times D^{(3)})$ is irreducible due to lemma \ref{lem:time_irred},
		so must be the field algebra $\A F(D_q\times D_c)$.
	\end{proof}

From a physical point of view, the above result is not surprising, and indeed quite
expected, for the actual supports of the test functions associated with states of
optimal localization are unbounded in the time direction. In fact, a similar behavior
was discovered to hold even in the case of the basic model, by Fredenhagen and Doplicher
in a unpublished joint work. We refrain from giving the proof of this result here,
for it is a simpler variant of the one given in this paper.

We conclude this section by briefly discussing how to index the above net by projections
in the Borel completion of the \calg\ describing the scale-covariant model of quantum
space-time, by merely applying the construction described in \cite{doplicher09}. To
any state $\omega\in\Omega(\{x_q\}\times D)$ (see equations \eqref{eq:states} for the definition
of $\Omega(\{x_q\}\times D)$) we associate the support projection $E_{\{x_q\}\times D}$
of the normal extension $\tilde\omega$ of $\omega$ to the the Borel completion of
the \calg, i.e. the minimal projection such that $\tilde\omega(E_{\{x_q\}\times D})=1$. 
Hence we can replace the set of all the regions of the form $\{x_q\}\times D$, i.e.
the image of the map $\pi$ defined by \eqref{eq:special_net}, with the family of these
projections.
Covariance then propagates from states, namely 
	\begin{equation}
		(g^{-1}\omega)(E_{\pi(gR)}) = 1,\qquad \forall \omega\in\Omega(\pi(R)),
	\end{equation}
for any $R\in\mathcal R$ (cf eq. \eqref{eq:regions}) and $g\in\mathscr G$, whence
\begin{equation}
	gE_{\pi(gR)} = E_{\pi(R)}.
\end{equation}
For a more detailed discussion we refer the reader to \cite{tornetta2011}.

	\section{Conclusions and outlook\label{chap: conclusions}}

In this paper we have investigated the quantum structure of a scale-covariant model
of quantum space-time. Quantum field theory on this model has provided a tool that
allowed us to carry out such an analysis. We conclude this paper by showing how the
massless scalar neutral field introduced throughout the exposition can be further
exploited in order to draw other insightful conclusions.

\subsection{Residual non-commutativity and CFTs in low dimensions}

The reader will recall that, on the scale-covariant model, two out of the four coordinates
of space-time are some of the generators of the center of the \calg. It is clear from
the analysis above that one of the two non-commuting coordinates is always time-like. We
might then restrict our attention to the non-singular part of the joint spectrum
of the commutators to obtain a model for a two-dimensional space-time. One
can then apply the results derived in this paper to the well-developed two-dimensional
theory of conformal fields and consider a free massless scalar neutral field $\phi$.
It is easily seen that it is possible to distinguish between the left and the
right chiral observables, for the compactification of the $1+1$ space-time splits into the product of two independent circles $S^1$. One may use the chiral fields to define nets of local
von Neumann field algebras on $S^1$, i.e. \cite{longo08, rehren01}
	\begin{equation}\label{eq:chiral_nets}
		I\mapsto\A A_L(I),\qquad J\mapsto\A A_R(J),
	\end{equation}
where $I$ and $J$ are any arcs in $S^1$. In general terms, if $I\mapsto\A A(I)$ is
a net of local von Neumann algebras on $S^1$ endowed with a strongly continuous unitary
representation $U_\phi$ of $PSL(2,\IR)$ subjected to the positive-energy spectrum
condition such that $(\Ad U_\phi)\A A(I) = \A A(\phi I)$, one may prove the \emph{Reeh-Schlieder}
property, i.e. the vacuum vector $\Omega_0$ is cyclic and separating for any arc
$I\subset S^1$. Besides, such nets satisfy the property of locality as well, i.e.
$I_1\cap I_2=\varnothing$ implies $\A A(I_1)\subset {\A A(I_2)}'$, provided that
$U$ is irreducible \cite{longo08}.

A net of local von Neumann algebras on the two-dimensional Minkowski space-time may
be constructed by taking the tensor product of the two nets \eqref{eq:chiral_nets}.
If $I,J$ are any two arcs in $S^1$, we define
	\begin{equation}
		I \times J := \{(x^0, x^1)\in\IR^2\ |\ u_-(x^0, x^1) \in I, u_+(x^0, x^1) \in J\},
	\end{equation}
and then set
	\begin{equation}
		O := I \times J,
	\end{equation}
which is a double cone in $\IR^2$. To $O$ we then associate the tensor product algebra given by
	\begin{equation}
		O\mapsto\A A_L(I)\otimes\A A_R(J),
	\end{equation}
and in this way we end up with a net of local algebras indexed by double cones in $\IR^2$.

To conclude this subsection we stress that, since the pair of coordinates we are left with do
not commute, we cannot use them to operate the usual construction of the chiral coordinates,
and obtain the very same local theory, for this would eventually yield a non-local
net of field algebras, as a consequence of what has been shown in this paper about
the non-locality of the scalar field on the scale-covariant model.

\subsection{Deformations of Quantum Spacetime and the full conformally-covariant model}

In this paper we limited our main discussion to the action
of the Poincar\'e and dilations groups, without going any further towards a full conformal symmetry. We will now argue about the reasons of such a limitation.

From the analysis of the quantum conditions in section \ref{sec:sc_quantum_cond}
it follows that the joint spectrum $\Sigma_0$ must satisfy almost the same constraints of the
electromagnetic tensor $F$ in the vacuum. Since the Maxwell equations are covariant
under the full conformal group (see \cite{bateman10} and references therein), it seems
legitimate to assume that the same behavior must be expected from the scale-covariant
model of quantum space-time. The problem to solve is then to find a way of defining the
relativistic ray inversion, i.e. the involutive diffeomorphism
	\begin{equation}\label{eq:inversion_diff}
		I(x)^\mu := \frac{x^\mu}{x^2},\qquad x^2=\eta_{\mu\nu}x^\mu x^\nu,
	\end{equation}
on the scale-covariant model of quantum space-time, since a generic special conformal
transformation is just the composition of a translation, an inversion and another
translation (see e.g. \cite{mirman05}). A moment's thought shows that there are no
trivial ways of defining a quantum equivalent of such a transformation, for the main
obstruction in this direction is the the ordering of the operators in generic functions of the coordinates, like the diffeomorphism given above.

We shall neglect this difficulty for the moment, and try to determine
a plausible automorphism on the \calg\ of the scale-covariant model of quantum space-time by exploiting the results on the free massless scalar neutral
field obtained in this paper along with the work of Hislop and Longo \cite[§2]{longo82}. For the sought
automorphism to be somehow linked to the diffeomorphism \eqref{eq:inversion_diff}
on classical space-time, we require the expression of $\phi(\omega)$ in terms of the
function $\psi_\omega(x)$ associated with $\omega$, namely
	\begin{equation}
		\phi(\omega)=\int\phi(x)\psi_\omega(x)\de^4x,
	\end{equation}
to behave as prescribed in \cite[§2]{longo82}. Thus we may tentatively introduce
the action $\rho$ on $\psi_\omega$ given by
	\begin{equation}
		(\rho\psi_\omega)(x) := \frac1{(x^2)^4}\psi_\omega\left(\frac x{x^2}\right),
	\end{equation}
but we cannot hope to choose $\omega$, even among the pure states, in order to avoid the singularity of \eqref{eq:inversion_diff} on the light-cone through the origin, for we see from
\eqref{eq:psi_scale_cov} that the support of $\psi_\omega$ is non-vanishing on the
boundary of the light-cone.

A way to overcome these obstructions seems to be the introduction of new generators of
the Lie algebra, i.e. to consider a deformation of the Weyl's algebra (cf. \cite[sect. 2]{doplicher96}). How arbitrarily
this can be done, however, depends on the requirement that the \emph{deformed}
model should lead to the desired uncertainty relations among the coordinates.

A second approach might consist of the quantization of the conformal compactification of
the classical Minkowski space-time (see e.g. \cite{jadczyk11} and references therein),
which in our view is not too different from the deformation approach,  for the compactified
Minkowski manifold requires the introduction of two additional conformal coordinates,
i.e. two new generators for the Lie algebra. The quantization of such a model would then proceed by considering the new
conformal coordinates, along with all the conformal invariants that may be constructed
with them, which are essential to define the quantum condition for the model.

Since it appears that a full conformally-covariant model requires an entirely different approach from the one used in this paper, the construction and the analysis of such a model will be pursued elsewhere.

	\section*{Acknowledgements}

Both the authors are deeply grateful to Professor Sergio Doplicher for proposing the idea of conducting a more thorough analysis of the scale-covariant model of quantum space-time and for the concrete opportunity to work on the subject. His constant help and abundant considerateness have also been crucial and must be acknowledged. The first named author wishes to thank Professor Sebastiano Carpi and Dr. Gerardo Morsella for their encouragement and technical help. Moreover, the second named author wishes to express his gratitude to Dr. Marcel Bischoff for his
technical support and fruitful comments on the topic of conformal field theories in lower dimensions.

	\appendix
	\section{Von Neumann algebras and real Hilbert subspaces\label{app:vn}}

In this appendix some basic definitions and results from the theory of von Neumann algebras and real subspaces of Hilbert spaces are collected for the reader's sake. Use of these tools has been made in the analysis of the von Neumann field algebras associated with the free scale-covariant scalar neutral field discussed in section \ref{sec:net}.

Let $\Hilb$ be a Hilbert space, and let $K\subset\Hilb$ be a closed real vector subspace of $\Hilb$. If $K \cap iK={0}$ then $K$ is said to be \emph{separating}; if $\overline{K+iK}=\Hilb$, then $K$ is said to be \emph{cyclic}.
\begin{definition}[Standard subspace] A closed real subspace $K$ of a Hilbert space $\Hilb$ is said to be \emph{standard} if it is both cyclic and separating.\end{definition}

We shall also recall the definition of \emph{standard von Neumann algebras}.
\begin{definition}[Standard von Neumann algebra] Let $\A R$ be a concrete von Neumann algebra on the Hilbert space $\Hilb$, and let $\Omega\in\Hilb$. If $\Omega$ is both cyclic and separating for $\A R$, then $\A R$ is said to be \emph{standard} with respect to $\Omega$.
\end{definition}
If $\A R$ and $\Omega$ are as in the above definition, then $(\A R,\Omega)$ is called a \emph{standard pair}.


\subsection{Local structure of the one-particle Hilbert space}

Let $E:\Sch(\IR^4,\IR) \to \Hilb$ be the embedding given by \cite{guido08,reed75}
\begin{equation}\label{eq:embedding}
	(Ef)(\bv p) := \sqrt{2\pi}\hat f(\sqrt{m^2 + \norm{\bv p}^2},\bv p),
\end{equation}
i.e. the restriction of the Fourier transform of $f\in\Sch(\IR^4,\IR)$ to the hyperboloid $\Omega_m^+$ of mass $m$ contained in the future light-cone, $\Hilb$ denoting the one-particle Hilbert space, i.e. $\Hilb=L^2(\Omega_m^+,\de\Omega_m^+)$. To each open and bounded region $O \subset \IR^4$ we may associate a closed real Hilbert subspace of $\Hilb$, namely \cite{araki64, guido08}
\begin{equation}\label{eq:net_rhs}
	K(O) := \overline{\{Ef\ |\ f\in\D(O,\IR)\}}.
\end{equation}
\begin{definition}[Causal complement] Let $O\subset\IR^4$ be an open region of Minkowski space-time. The \emph{causal complement} of $O$ is the open region $O'\subset\IR^4$ given by
	$$O' := \{x\in\IR^4\ |\ x-y\text{ is space-like for any }y\in O\}.$$
\end{definition}
Among all the regions in $\IR^4$, those that satisfies the property $O=O''$ are somewhat special and are called \emph{causally complete}. Some notable examples of such regions are the families of wedges $\mathcal W$ and double cones $\mathcal K$ \cite{haag96}.

\begin{definition}[Symplectic complement] Let $K$ be a closed real Hilbert subspace of the Hilbert space $\Hilb$. The \emph{symplectic complement} of $K$ is the vector subspace $K'\subset\Hilb$ given by
	$$K' := \{f\in\Hilb\ |\ \Im(f,h) = 0\quad\forall h\in K\}.$$
\end{definition}

A result that may be found in \cite{araki64} (cf. also \cite{guido08}) establishes some interesting properties concerning the net of real Hilbert subspace \eqref{eq:net_rhs}, namely

\begin{theorem} \label{thm:araki} Let $O$ be any non-empty simply connected open region of $\IR^4$ with regular boundary. Then
	\begin{itemize}
		\item $K(O')=K(O)'$;
		\item the real Hilbert subspace $K(O)$ of $\Hilb$ is standard.
	\end{itemize}
\end{theorem}

\subsection{Second quantization and field algebras}

Let $\Hilb$ be the one-particle Hilbert space, and let $\Hilb_F=e^\Hilb$ be the bosonic Fock Hilbert space constructed from $\Hilb$. A \emph{coherent vector} is the element of $\Hilb_F$ given by \cite{guido08}
\begin{equation}
	e^{f}:=\bigoplus_{n=0}^\infty\frac{f^{\otimes n}}{\sqrt{n!}},\qquad f\in\Hilb.
\end{equation}
A \emph{Weyl operator} is a map $W:\Hilb\to\Bound(\Hilb_F)$ such that \cite{guido08, bogoliubov90}
\begin{subequations}
\begin{align}
	W(f)\Omega &= e^{-\frac12\norm f^2}e^{if},\qquad\forall f\in\Hilb\\
	W(f)W(g) &= e^{-\frac12\Im(f,g)}W(f+g),\qquad\forall f,g\in\Hilb,
\end{align}
\end{subequations}
where $\Omega\in\Hilb_F$ denotes the Fock vacuum vector. It can be shown \cite{guido08} that the system of coherent vectors is total in $\Hilb_F$, and hence each $W(f)$, $f\in\Hilb$ extends uniquely to a unitary operator on $\Hilb_F$.

Let now $K\subset\Hilb$ be a closed real Hilbert subspace of $\Hilb$. The \emph{von Neumann field algebra} associated with $K$ is the algebra $\A R(K)$ generated by
\begin{equation}\label{eq:vn_net}
	\A R(K):=\{W(f)\ |\ f\in K\}''.
\end{equation}
A couple of useful results concerning von Neumann field algebras associated with real Hilbert subspace are the following \cite{guido08}
\begin{lemma} \label{lem:R_of_K} Let $K\subset\Hilb$ be a real Hilbert subspace of the one-particle Hilbert space $\Hilb$. The vacuum vector $\Omega\in\Hilb_F$ is
\begin{enumerate}
	\item separating for $\A R(K)$ iff $K$ is separating;
	\item cyclic for $\A R(K)$ iff $K$ is cyclic.
\end{enumerate}
\end{lemma}\label{lem:araki_haag_duality}
\begin{lemma} Let $K$ be a standard subspace of the Hilbert space $\Hilb$. Then the von Neumann algebra $\A R(K)$ associated with $K$ is such that $\A R(K')=\A R(K)'$.
\end{lemma}
Let $O\mapsto K(O)$, $O\subset\IR^4$, be a net of closed real Hilbert subspaces of $\Hilb$, and let $\A R(K)$ be the net of von Neumann algebras defined in \eqref{eq:vn_net}. A new net of von Neumann algebras, indexed by open regions of $\IR^4$ may then be constructed by
\begin{subequations}\label{eq:local_field_algebra}
	\begin{equation}
		O\mapsto \A R(O),
	\end{equation}
where
	\begin{equation}
		\A R(O) := \A R(K(O)).
	\end{equation}
\end{subequations}
If we restrict our attention to the family of double cones $\mathcal K$ in $\IR^4$, then the so-called \emph{Haag-duality} holds \cite{haag96}, namely
	\begin{equation}
		\A R (K') = \A R(K)',\qquad \forall K\in\mathcal K,
	\end{equation}
which is an extension of the causality requirement
	$$\A R(O')\subset \A R(O)'$$
for any open region $O\subset\IR^4$.

	\bibliography{refs}

\begin{thebibliography}{24}%
\makeatletter
\providecommand \@ifxundefined [1]{%
 \@ifx{#1\undefined}
}%
\providecommand \@ifnum [1]{%
 \ifnum #1\expandafter \@firstoftwo
 \else \expandafter \@secondoftwo
 \fi
}%
\providecommand \@ifx [1]{%
 \ifx #1\expandafter \@firstoftwo
 \else \expandafter \@secondoftwo
 \fi
}%
\providecommand \natexlab [1]{#1}%
\providecommand \enquote  [1]{``#1''}%
\providecommand \bibnamefont  [1]{#1}%
\providecommand \bibfnamefont [1]{#1}%
\providecommand \citenamefont [1]{#1}%
\providecommand \href@noop [0]{\@secondoftwo}%
\providecommand \href [0]{\begingroup \@sanitize@url \@href}%
\providecommand \@href[1]{\@@startlink{#1}\@@href}%
\providecommand \@@href[1]{\endgroup#1\@@endlink}%
\providecommand \@sanitize@url [0]{\catcode `\\12\catcode `\$12\catcode
  `\&12\catcode `\#12\catcode `\^12\catcode `\_12\catcode `\%12\relax}%
\providecommand \@@startlink[1]{}%
\providecommand \@@endlink[0]{}%
\providecommand \url  [0]{\begingroup\@sanitize@url \@url }%
\providecommand \@url [1]{\endgroup\@href {#1}{\urlprefix }}%
\providecommand \urlprefix  [0]{URL }%
\providecommand \Eprint [0]{\href }%
\providecommand \doibase [0]{http://dx.doi.org/}%
\providecommand \selectlanguage [0]{\@gobble}%
\providecommand \bibinfo  [0]{\@secondoftwo}%
\providecommand \bibfield  [0]{\@secondoftwo}%
\providecommand \translation [1]{[#1]}%
\providecommand \BibitemOpen [0]{}%
\providecommand \bibitemStop [0]{}%
\providecommand \bibitemNoStop [0]{.\EOS\space}%
\providecommand \EOS [0]{\spacefactor3000\relax}%
\providecommand \BibitemShut  [1]{\csname bibitem#1\endcsname}%
\let\auto@bib@innerbib\@empty
\bibitem [{\citenamefont {Doplicher}\ \emph {et~al.}({1995})\citenamefont
  {Doplicher}, \citenamefont {Fredenhagen},\ and\ \citenamefont
  {Roberts}}]{doplicher95}%
  \BibitemOpen
  \bibfield  {author} {\bibinfo {author} {\bibfnamefont {Sergio}\ \bibnamefont
  {Doplicher}}, \bibinfo {author} {\bibfnamefont {Klaus}\ \bibnamefont
  {Fredenhagen}}, \ and\ \bibinfo {author} {\bibfnamefont {John~E.}\
  \bibnamefont {Roberts}},\ }\bibfield  {title} {\enquote {\bibinfo {title}
  {{The quantum structure of spacetime at the Planck scale and quantum
  fields}},}\ }\href@noop {} {\bibfield  {journal} {\bibinfo  {journal}
  {{Commun. Math. Phys.}}\ }\textbf {\bibinfo {volume} {{172}}},\ \bibinfo
  {pages} {187--220} (\bibinfo {year} {{1995}})},\ \Eprint
  {http://arxiv.org/abs/{hep-th/0303037v1}} {{hep-th/0303037v1}} \BibitemShut
  {NoStop}%
\bibitem [{\citenamefont {Wigner}(1957)}]{Wigner:1957ep}%
  \BibitemOpen
  \bibfield  {author} {\bibinfo {author} {\bibfnamefont {Eugene~P.}\
  \bibnamefont {Wigner}},\ }\bibfield  {title} {\enquote {\bibinfo {title}
  {{Relativistic Invariance and Quantum Phenomena}},}\ }\href {\doibase
  10.1103/RevModPhys.29.255} {\bibfield  {journal} {\bibinfo  {journal} {Rev.
  Mod. Phys.}\ }\textbf {\bibinfo {volume} {29}},\ \bibinfo {pages} {255--268}
  (\bibinfo {year} {1957})}\BibitemShut {NoStop}%
\bibitem [{\citenamefont {Salecker}\ and\ \citenamefont
  {Wigner}(1958)}]{Salecker:1957be}%
  \BibitemOpen
  \bibfield  {author} {\bibinfo {author} {\bibfnamefont {H.}~\bibnamefont
  {Salecker}}\ and\ \bibinfo {author} {\bibfnamefont {E.P.}\ \bibnamefont
  {Wigner}},\ }\bibfield  {title} {\enquote {\bibinfo {title} {{Quantum
  limitations of the measurement of space-time distances}},}\ }\href {\doibase
  10.1103/PhysRev.109.571} {\bibfield  {journal} {\bibinfo  {journal} {Phys.
  Rev.}\ }\textbf {\bibinfo {volume} {109}},\ \bibinfo {pages} {571--577}
  (\bibinfo {year} {1958})}\BibitemShut {NoStop}%
\bibitem [{\citenamefont {Rovelli}({2004})}]{rovelli04}%
  \BibitemOpen
  \bibfield  {author} {\bibinfo {author} {\bibfnamefont {Carlo}\ \bibnamefont
  {Rovelli}},\ }\href@noop {} {\emph {\bibinfo {title} {{Quantum gravity}}}}\
  (\bibinfo  {publisher} {{Cambridge University Press}},\ \bibinfo {year}
  {{2004}})\BibitemShut {NoStop}%
\bibitem [{\citenamefont {Oriti}({2009})}]{oriti09}%
  \BibitemOpen
  \bibfield  {author} {\bibinfo {author} {\bibfnamefont {Daniele}\ \bibnamefont
  {Oriti}},\ }\href@noop {} {\emph {\bibinfo {title} {{Approaches to Quantum
  Gravity}}}}\ (\bibinfo  {publisher} {{Cambridge University Press}},\ \bibinfo
  {year} {{2009}})\BibitemShut {NoStop}%
\bibitem [{\citenamefont {Piacitelli}({2010})}]{piacitelli11}%
  \BibitemOpen
  \bibfield  {author} {\bibinfo {author} {\bibfnamefont {Gherardo}\
  \bibnamefont {Piacitelli}},\ }\bibfield  {title} {\enquote {\bibinfo {title}
  {{Aspects of Quantum Field Theory on Quantum Spacetime}},}\ }\href@noop {}
  {\bibfield  {journal} {\bibinfo  {journal} {{PoS(CNCFG2010)027}}\ } (\bibinfo
  {year} {{2010}})},\ \Eprint {http://arxiv.org/abs/{arXiv:1103.3405v1
  [hep-th]}} {{arXiv:1103.3405v1 [hep-th]}} \BibitemShut {NoStop}%
\bibitem [{\citenamefont {Haag}({1996})}]{haag96}%
  \BibitemOpen
  \bibfield  {author} {\bibinfo {author} {\bibfnamefont {Rudolf}\ \bibnamefont
  {Haag}},\ }\href@noop {} {\emph {\bibinfo {title} {{Local Quantum Physics:
  Fields, Particles, Algebras}}}},\ \bibinfo {edition} {{Second Revised and
  Enlarged}}\ ed.,\ {Theoretical and Mathematical Physics}\ (\bibinfo
  {publisher} {{Springer}},\ \bibinfo {year} {{1996}})\BibitemShut {NoStop}%
\bibitem [{\citenamefont {Araki}({2000})}]{araki00}%
  \BibitemOpen
  \bibfield  {author} {\bibinfo {author} {\bibfnamefont {Huzihiro}\
  \bibnamefont {Araki}},\ }\href@noop {} {\emph {\bibinfo {title}
  {{Mathematical Theory of Quantum Fields}}}},\ {International Series of
  Monographs on Physics}\ (\bibinfo  {publisher} {{Oxford University Press}},\
  \bibinfo {year} {{2000}})\BibitemShut {NoStop}%
\bibitem [{\citenamefont {Doplicher}({2010})}]{doplicher09}%
  \BibitemOpen
  \bibfield  {author} {\bibinfo {author} {\bibfnamefont {Sergio}\ \bibnamefont
  {Doplicher}},\ }\bibfield  {title} {\enquote {\bibinfo {title} {{The
  Principle of Locality. Effectiveness, fate and challenges}},}\ }\href
  {\doibase {10.1063/1.3276100}} {\bibfield  {journal} {\bibinfo  {journal}
  {{J. Math. Phys.}}\ }\textbf {\bibinfo {volume} {{51}}},\ \bibinfo {pages}
  {20} (\bibinfo {year} {{2010}})}\BibitemShut {NoStop}%
\bibitem [{\citenamefont {Perini}({2006})}]{perini}%
  \BibitemOpen
  \bibfield  {author} {\bibinfo {author} {\bibfnamefont {Claudio}\ \bibnamefont
  {Perini}},\ }\emph {\bibinfo {title} {{Un modello quantistico di spazio-tempo
  senza lunghezza fondamentale}}},\ \href@noop {} {Master's thesis},\ \bibinfo
  {school} {{Universit{\`{a}} degli studi Roma Tre}} (\bibinfo {year}
  {{2006}})\BibitemShut {NoStop}%
\bibitem [{\citenamefont {Di~Francesco}\ \emph {et~al.}({1996})\citenamefont
  {Di~Francesco}, \citenamefont {Mathieu},\ and\ \citenamefont
  {S{\'{e}}n{\'{e}}chal}}]{difrancesco96}%
  \BibitemOpen
  \bibfield  {author} {\bibinfo {author} {\bibfnamefont {Philippe}\
  \bibnamefont {Di~Francesco}}, \bibinfo {author} {\bibfnamefont {Pierre}\
  \bibnamefont {Mathieu}}, \ and\ \bibinfo {author} {\bibfnamefont {David}\
  \bibnamefont {S{\'{e}}n{\'{e}}chal}},\ }\href@noop {} {\emph {\bibinfo
  {title} {{Conformal Field Theory}}}},\ \bibinfo {edition} {{Corrected}}\
  ed.,\ {Graduate Texts in Contemporary Physics}\ (\bibinfo  {publisher}
  {{Springer}},\ \bibinfo {year} {{1996}})\BibitemShut {NoStop}%
\bibitem [{\citenamefont {Streater}\ and\ \citenamefont
  {Wightman}({1964})}]{wightman64}%
  \BibitemOpen
  \bibfield  {author} {\bibinfo {author} {\bibfnamefont {Raymond~F.}\
  \bibnamefont {Streater}}\ and\ \bibinfo {author} {\bibfnamefont {Arthur~S.}\
  \bibnamefont {Wightman}},\ }\href@noop {} {\emph {\bibinfo {title} {{PCT,
  Spin and Statistics, and All That}}}}\ (\bibinfo  {publisher} {{W. A.
  Benjamin, Inc.}},\ \bibinfo {year} {{1964}})\BibitemShut {NoStop}%
\bibitem [{\citenamefont {Araki}({1964})}]{araki64}%
  \BibitemOpen
  \bibfield  {author} {\bibinfo {author} {\bibfnamefont {Huzihiro}\
  \bibnamefont {Araki}},\ }\bibfield  {title} {\enquote {\bibinfo {title} {{Von
  Neumann algebras of local observables for free scalar field}},}\ }\href@noop
  {} {\bibfield  {journal} {\bibinfo  {journal} {{J. Math. Phys.}}\ }\textbf
  {\bibinfo {volume} {{4}}},\ \bibinfo {pages} {1343--1362} (\bibinfo {year}
  {{1964}})}\BibitemShut {NoStop}%
\bibitem [{\citenamefont {Tornetta}({2011})}]{tornetta2011}%
  \BibitemOpen
  \bibfield  {author} {\bibinfo {author} {\bibfnamefont {Gabriele~N.}\
  \bibnamefont {Tornetta}},\ }\emph {\bibinfo {title} {{Scale-covariant Field
  Algebras on a Quantum Space-time Model}}},\ \href@noop {} {Master's thesis},\
  \bibinfo  {school} {{Department of Physics, Universit{\`{a}} degli Studi di
  Roma ``La Sapienza''}}, \bibinfo {address} {{Piazzale Aldo Moro 5, 00185
  Roma}} (\bibinfo {year} {{2011}})\BibitemShut {NoStop}%
\bibitem [{\citenamefont {Longo}({2008})}]{longo08}%
  \BibitemOpen
  \bibfield  {author} {\bibinfo {author} {\bibfnamefont {Roberto}\ \bibnamefont
  {Longo}},\ }\bibfield  {title} {\enquote {\bibinfo {title} {{Real Hilbert
  subspaces, modular theory, $SL(2,\mathbb{R})$ and CFT}},}\ }in\ \href@noop {}
  {\emph {\bibinfo {booktitle} {{Von Neumann Algebras in Sibiu: Conference
  Proceedings, Sibiu, June 9-16, 2007}}}},\ \bibinfo {editor} {edited by\
  \bibinfo {editor} {\bibfnamefont {Ken}\ \bibnamefont {Dykema}}\ and\ \bibinfo
  {editor} {\bibfnamefont {Florin}\ \bibnamefont {R{\u{a}}dulescu}}}\ (\bibinfo
   {publisher} {{Theta Foundation}},\ \bibinfo {year} {{2008}})\BibitemShut
  {NoStop}%
\bibitem [{\citenamefont {Rehren}({2001})}]{rehren01}%
  \BibitemOpen
  \bibfield  {author} {\bibinfo {author} {\bibfnamefont {Karl-Henning}\
  \bibnamefont {Rehren}},\ }\bibfield  {title} {\enquote {\bibinfo {title}
  {{Locality and modular invariance in 2D conformal QFT}},}\ }\href@noop {}
  {\bibfield  {journal} {\bibinfo  {journal} {{Fields Inst. Commun.}}\ }\textbf
  {\bibinfo {volume} {{30}}},\ \bibinfo {pages} {341} (\bibinfo {year}
  {{2001}})},\ \Eprint {http://arxiv.org/abs/{math-ph/0009004v1}}
  {{math-ph/0009004v1}} \BibitemShut {NoStop}%
\bibitem [{\citenamefont {Bateman}({1910})}]{bateman10}%
  \BibitemOpen
  \bibfield  {author} {\bibinfo {author} {\bibfnamefont {Harry}\ \bibnamefont
  {Bateman}},\ }\bibfield  {title} {\enquote {\bibinfo {title} {{The
  transformation of the electrodynamical equations}},}\ }\href@noop {}
  {\bibfield  {journal} {\bibinfo  {journal} {{Proc. London Math. Soc.}}\
  }\textbf {\bibinfo {volume} {{8}}},\ \bibinfo {pages} {223{--}264} (\bibinfo
  {year} {{1910}})}\BibitemShut {NoStop}%
\bibitem [{\citenamefont {Mirman}({2005})}]{mirman05}%
  \BibitemOpen
  \bibfield  {author} {\bibinfo {author} {\bibfnamefont {Ronald}\ \bibnamefont
  {Mirman}},\ }\href@noop {} {\emph {\bibinfo {title} {{Quantum Field Theory
  Conformal Group Theory Conformal Field Theory: Mathematical and Conceptual
  Foundations Physical and Geometrical Applications}}}}\ (\bibinfo  {publisher}
  {{Backinprint.com}},\ \bibinfo {year} {{2005}})\BibitemShut {NoStop}%
\bibitem [{\citenamefont {Hislop}\ and\ \citenamefont
  {Longo}({1982})}]{longo82}%
  \BibitemOpen
  \bibfield  {author} {\bibinfo {author} {\bibfnamefont {Peter~D.}\
  \bibnamefont {Hislop}}\ and\ \bibinfo {author} {\bibfnamefont {Roberto}\
  \bibnamefont {Longo}},\ }\bibfield  {title} {\enquote {\bibinfo {title}
  {{Modular Structure of the Local Algebras Associated with the Free Massless
  Scalar Field Theory}},}\ }\href@noop {} {\bibfield  {journal} {\bibinfo
  {journal} {{Commun. Math. Phys.}}\ }\textbf {\bibinfo {volume} {{84}}},\
  \bibinfo {pages} {71--85} (\bibinfo {year} {{1982}})}\BibitemShut {NoStop}%
\bibitem [{\citenamefont {Doplicher}({1996})}]{doplicher96}%
  \BibitemOpen
  \bibfield  {author} {\bibinfo {author} {\bibfnamefont {Sergio}\ \bibnamefont
  {Doplicher}},\ }\bibfield  {title} {\enquote {\bibinfo {title} {{Quantum
  spacetime}},}\ }\href@noop {} {\bibfield  {journal} {\bibinfo  {journal}
  {{Ann. Inst. H. Poincar{\'{e}} Phys. Th{\'{e}}or.}}\ }\textbf {\bibinfo
  {volume} {{64}}},\ \bibinfo {pages} {543--553} (\bibinfo {year} {{1996}})},\
  \bibinfo {note} {{New problems in the general theory of fields and particles,
  Part II (Paris, 1994)}}\BibitemShut {NoStop}%
\bibitem [{\citenamefont {Jadczyk}({2011})}]{jadczyk11}%
  \BibitemOpen
  \bibfield  {author} {\bibinfo {author} {\bibfnamefont {Arkadiusz}\
  \bibnamefont {Jadczyk}},\ }\bibfield  {title} {\enquote {\bibinfo {title}
  {{Compactified Minkowski Space: Myths and Facts}},}\ }\href@noop {} {\
  (\bibinfo {year} {{2011}})},\ \Eprint
  {http://arxiv.org/abs/{arXiv:1105.3948v2 [math-ph]}} {{arXiv:1105.3948v2
  [math-ph]}} \BibitemShut {NoStop}%
\bibitem [{\citenamefont {Guido}({2011})}]{guido08}%
  \BibitemOpen
  \bibfield  {author} {\bibinfo {author} {\bibfnamefont {Daniele}\ \bibnamefont
  {Guido}},\ }\bibfield  {title} {\enquote {\bibinfo {title} {{Modular theory
  for the von Neumann algebras of Local Quantum Physics}},}\ }\href@noop {}
  {\bibfield  {journal} {\bibinfo  {journal} {{Contemp. Math.}}\ }\textbf
  {\bibinfo {volume} {{534}}},\ \bibinfo {pages} {97--120} (\bibinfo {year}
  {{2011}})},\ \Eprint {http://arxiv.org/abs/{arXiv:0812.1511v1 [math.OA]}}
  {{arXiv:0812.1511v1 [math.OA]}} \BibitemShut {NoStop}%
\bibitem [{\citenamefont {Reed}\ and\ \citenamefont {Simon}({1975})}]{reed75}%
  \BibitemOpen
  \bibfield  {author} {\bibinfo {author} {\bibfnamefont {Michael}\ \bibnamefont
  {Reed}}\ and\ \bibinfo {author} {\bibfnamefont {Barry}\ \bibnamefont
  {Simon}},\ }\href@noop {} {\emph {\bibinfo {title} {{Methods of Modern
  Mathematical Physics}}}},\ Vol.\ \bibinfo {volume} {{II: Fourier Analysis,
  Self-Adjointness}}\ (\bibinfo  {publisher} {{Academic Press}},\ \bibinfo
  {year} {{1975}})\BibitemShut {NoStop}%
\bibitem [{\citenamefont {Bogolubov}\ \emph {et~al.}({1990})\citenamefont
  {Bogolubov}, \citenamefont {Logunov}, \citenamefont {Oksak},\ and\
  \citenamefont {Todorov}}]{bogoliubov90}%
  \BibitemOpen
  \bibfield  {author} {\bibinfo {author} {\bibfnamefont {Nicolay~N.}\
  \bibnamefont {Bogolubov}}, \bibinfo {author} {\bibfnamefont {Anatoli~A.}\
  \bibnamefont {Logunov}}, \bibinfo {author} {\bibfnamefont {A.~I.}\
  \bibnamefont {Oksak}}, \ and\ \bibinfo {author} {\bibfnamefont {Ivan~T.}\
  \bibnamefont {Todorov}},\ }\href@noop {} {\emph {\bibinfo {title} {{General
  Principles of Quantum Field Theory}}}},\ {Mathematical physics and applied
  mathematics}\ (\bibinfo  {publisher} {{Kluwer Academic Publishers}},\
  \bibinfo {year} {{1990}})\BibitemShut {NoStop}%
\end{thebibliography}%
	
\end{document}